\documentclass[12pt]{article}

\usepackage{amssymb} 
\usepackage{amsfonts}
\usepackage{amsmath,wasysym}

\usepackage{amsthm,amsbsy,color}

\usepackage{geometry} 
\usepackage{setspace}
\usepackage[bottom]{footmisc}
\usepackage{indentfirst}
\usepackage{multirow}
\usepackage{booktabs}
\usepackage{tabularx}
\usepackage{float}
\usepackage{tabulary}
\usepackage{booktabs}
\usepackage{rotating}
\usepackage{placeins}
\usepackage{floatflt}
\usepackage[final]{pdfpages}
\usepackage{lscape}
\usepackage{pdflscape}
\usepackage{caption}
\usepackage{subcaption}
\usepackage{verbatim}
\usepackage{graphicx}
\usepackage{epstopdf}
\usepackage{wrapfig}
\usepackage[round]{natbib}
\usepackage{url}
\usepackage{soul, color}
\usepackage{xr-hyper}
\usepackage[colorlinks=true,citecolor=blue,urlcolor=blue,linkcolor=black, hyperfootnotes=true]{hyperref}
\usepackage{bm}

\newcommand{\RR}{\mathbb{R}}

\newtheorem{theorem}{Theorem}

\newtheorem{example}{Example}

\DeclareMathOperator{\vect}{vec}
\newcommand{\argmin}{\operatornamewithlimits{argmin}}

\numberwithin{equation}{section}
\numberwithin{figure}{section}
\numberwithin{table}{section}

\newcolumntype{L}[1]{>{\arraybackslash}p{#1}}

\usepackage[title]{appendix}

\begin{document}

\thispagestyle{empty}
\begin{center}
\renewcommand{\thefootnote}{\fnsymbol{footnote}}
 \textbf{\large Estimation of a Factor-Augmented Linear Model with Applications Using Student Achievement Data\footnote{The authors would like to thank Irene Botosaru, Rajeev Darolia, Atsushi Inoue, Ana Herrera, Jeff Racine, and Youngki Shin for comments and useful conversations. We thank Chris Heaton for sharing his programs for estimation of the approximate factor model, and Jérôme Adda and Michele Pellizzari for sharing the Bocconi data.}
 ~\\ }
 \vspace{0.2in} Matthew Harding\footnote{Department of Economics, University of California at Irvine, SSPB 3207, Irvine, CA 92697; Email: \texttt{harding1@uci.edu}}, Carlos Lamarche\footnote{Department of Economics, University of Kentucky, 223G Gatton College of Business and Economics, Lexington, KY 40506-0034; Email: \texttt{clamarche@uky.edu}}, and Chris Muris\footnote{Department of Economics, McMaster University, Kenneth Taylor Hall 406, Hamilton, Ontario L8S 4M4; Email: \texttt{muerisc@mcmaster.ca}}
\end{center}
\begin{center}
\today \\
\end{center}



\vspace{.3in} \hrule   \noindent
\begin{center}
\textbf{Abstract}
\end{center}
\baselineskip=.88\baselineskip  {\small

\noindent In many longitudinal settings, economic theory does not guide practitioners on the type of restrictions that must be imposed to solve the rotational indeterminacy of factor-augmented linear models. We study this problem and offer several novel results on identification using internally generated instruments. We propose a new class of estimators and establish large sample results using recent developments on clustered samples and high-dimensional models. We carry out simulation studies which show that the proposed approaches improve the performance of existing methods on the estimation of unknown factors. Lastly, we consider three empirical applications using administrative data of students clustered in different subjects in elementary school, high school and college.}

\vspace{.25in}

\baselineskip=1.1\baselineskip
 \vspace{.1in} \hrule \vspace{.2in}
\setcounter{page}{0}
\setcounter{footnote}{0}

\clearpage

\section{Introduction}

In recent years, there has been an increasing interest in applications of factor models in economics, finance, and psychology. In economics, the identification and estimation of factor models has received substantial attention in a number of areas from macro-finance to labor economics and development \citep{bernanke2005measuring,kim2014divorce,oAttanasio2020}. Important work has studied the role of cognition, personality traits, and academic motivation on child development \citep{fCunha2007,fCunha2008,Borghans02102008,jHeckman13}. Factor-augmented regressions as in \citet{jStock1999,jStock2002} are known to improve forecasts of macroeconomic time series such as inflation and industrial production. The literature also includes new models for high-dimensional data sets \citep{jBai2016}, and methods for panels with large cross-sectional ($N$) and time-series ($T$) dimensions, following the influential work by \citet{hP06} and \citet{jB09}. In panel data econometrics, one popular interpretation treats the latent factors as a generalization of traditional fixed effects models \citep[among others]{harding2014estimating,Chudik2015393,martin2015,moon2017,tAndo2016,aJuodis2018,tAndo2017,mHarding2020}. 

While the estimators proposed for panels with large dimensions have been widely popular, other methods developed for panels with small, or fixed, $T$ have not been frequently adopted by practitioners conducting empirical academic research. One reason, as mentioned in \citet{aJuodis2020} and illustrated in \citet{oAttanasio2020} and \citet{DelBono2020}, is that identification of the factor model requires normalization restrictions that matter for the interpretation of results \citep{NBERw22441}. In some cases identification is achieved through the use of dedicated measurements, where \textit{a priori} knowledge is used to associate certain measurements uniquely with specific factors (for example a test can be associated uniquely with a given skill e.g. \citet{cunha2021econometrics}).   One common restriction, labeled ``PC3'' in \citet{BAI2013}, normalizes coefficients in the first block of factors. However, a large set of normalizations are available to practitioners when observed measurements per subject do not have a predetermined or natural order. In this paper, we investigate this problem while primarily focusing our analysis on ``fixed-$J$'' panels, where $J$ denotes a number of clusters or groups (e.g., states, counties, schools, etc. as opposed to time series).

We begin our investigation by introducing a class of estimators that use internally generated instruments. Papers by \citet{HEATON2012348}, and \citet{aJuodis2020}, among others, also propose to estimate similar  models using internally constructed instruments, an idea that can be traced back to the work of \citet{madansky1964instrumental}. In contrast to the existing literature, our model is identified based on an alternative non-singular transformation which includes PC3 as a special case. This normalization is convenient for the interpretation of results when economic theory is silent on the type of restrictions that must be imposed to solve the rotational indeterminacy of the factor model. Moreover, we establish large sample results by accommodating the asymptotic theory for clusters developed by \citet{bHansen2019}.

We then consider adopting multiple non-singular transformations to improve the efficiency of the estimator, and we derive two additional theoretical results on estimation. We propose an estimator considering PC3-type restrictions in fixed-$J$ panels and show that the estimator is consistent and asymptotically normal under standard conditions. However, in applications to high-dimensional data or panel data with a large number of clusters, there is an increasing number of available transformations. The number of instrumental variables can also increase with $J$, creating finite sample bias similar to the one generated by the use of too many instruments \citep[see][]{jHahn2003, Jerry2008,pBekker1994}. Our second development is to address poor finite sample performance by proposing an alternative two-step estimator that accommodates econometric methods for high-dimensional models in a first step \citep*[e.g.,][]{Belloni2012,linton2016,fWind2019}. Although our main focus is on fixed-$J$ panels, we establish the asymptotic distribution theory for multiple transformations and demonstrate to practitioners how to select normalizations out of (possibly) an infinite number of them. 

Despite the large body of work on instrumental variables and factor models, this paper develops a new class of estimators that are simple to implement and offer practitioners better performance in small samples. The estimation of slope parameters using instrumental variables is investigated in \citet{baing2010}, \citet{Harding2011197}, \citet{ahn2013panel}, \citet{Robertson2015526}, \citet{aJuodis2020}, and \citet{NORKUTE2021416}, among others. On the other hand, the latent factor structure is estimated in \citet{madansky1964instrumental}, \citet{GH1982}, \citet{sPudney81}, \citet{jHeckman1987}, and \citet{HEATON2012348}. In our simulation study, we find that the proposed estimators improve on the performance of existing instrumental variable methods for the estimation of unknown factors. 

Lastly, we consider three empirical applications of our method to the estimation of models of educational attainment using administrative data on students. First, we investigate how the distribution of students' abilities at a school district level changes over subsequent years of K12 education. We present evidence on the temporal and geographic variability of educational opportunity across the US using administrative data from over 11,000 school districts. In our second illustration of the approach, we estimate a factor model using administrative data from a higher education institution in Europe \citep{DeGiorgi2012}. The third application employs data from \citet{jAngrist2002} to evaluate the impact of an educational voucher program implemented in Latin America. These examples show intriguing results and highlight the usefulness of our techniques in varied settings in order to identify the strong and weak performers across the unobserved dimensions of academic achievement. 

This paper is organized as follows. Section 2 introduces the factor-augmented linear model and the proposed estimator. The section also presents the main theoretical result and discusses the implementation of the estimator. Section 3 investigates estimation under multiple normalizations. Section 4 provides Monte Carlo experiments to investigate the small sample performance of the proposed estimators. Section 5 demonstrates how the approaches can be used in practice by exploring applications using administrative data. Section 6 concludes. Mathematical proofs are offered in the Appendix. 

\section{Model and estimation}

This paper considers the following factor-augmented linear model for $i=1,\hdots,N$ subjects and $j=1,\hdots,J$ clusters: 
\begin{equation}
y_{ij}=\bm{x}_{ij}'\bm{\beta}+\bm{\lambda}_{i}'\bm{f}_{j}+u_{ij}, \label{panel}
\end{equation}
where $y_{ij} \in \RR$ is the $j$-th response variable for subject $i$, $\bm{x}_{ij} \in \RR^p$ is a vector of independent variables, $\bm{\beta} \in \mathcal{B} \subseteq \RR^p$ is an unknown parameter vector, $\bm{\lambda}_{i} = (\lambda_{i1},\lambda_{i2},\hdots,\lambda_{ir})' \in \RR^r$ is a vector of factor loadings, $\bm{f}_{j} = (f_{j1},f_{j2},\hdots,f_{jr})' \in \RR^r$ is a vector of latent factors, and $u_{ij}$ is an error term. The number of factors $r$ does not need to be known, as one can determine the number of factors following a number of approaches \citep[e.g.,][]{jB2002,aOnatski2010,gKapetanios2010,sAhn2013,lTrapani2018}.

We are interested in the estimation of $\bm{\beta}$ and $\bm{f}_{j}$. For the results in this section, we will fix a subset $A_0$ of groups of interest, and estimate $\left(\bm{f}_{j},j\in A_0\right)$. Throughout this section, even as $J$ diverges, this subset remains fixed. Once estimators of the factors and of $\bm{\beta}$ are available, it is straightforward to construct an estimator for $\bm{\lambda}_{i}$ \citep[see, e.g.,][]{HEATON2012348,BAI2013}. In Section \ref{empirical-application}, as an illustration of the approach, we first concentrate our attention on estimation of the factor $\bm{f}_j$, and then we estimate the loading $\bm{\lambda}_i$ for $1 \leq i \leq N$. 



Based on equation \eqref{panel}, consider 
\begin{equation}
\bm{y}_{i} = \left(\bm{y}_{iA_{0}}', \bm{y}_{iA_{J} \setminus A_0}', \bm{y}_{iB_{J}}' \right)',
\end{equation}
where $A_J$ is a set that includes groups that are used to proxy the vector of loadings $\bm{\lambda}_i$, and $B_J$ is a set that includes groups that are used to generate instrumental variables. The number of elements in each set $S$ is denoted by $m_S$, and we require $m_{A_J} \geq r$, and $m_{B_J} \geq r,$ and $\left(A_0 \cup A_J \right) \cap B_J = \emptyset$. We will also require that $A_J \cap A_0 = \emptyset$, although this can be relaxed at the cost of additional notation and subtleties. For instance, we could require $m_{A_J \setminus A_0} \geq r$, and that there is at least one $j \in A_J \setminus A_0$ involved in each of the $r$ averages discussed below. It follows that, 
\begin{equation}
\bm{y}_{i A_J}=\bm{x}_{iA_J}'\bm{\beta}+\bm{f}_{A_J}'\bm{\lambda}_{i}+\bm{u}_{iA_J}, \label{panel:AT}
\end{equation}
where $\bm{y}_{i A_J}$ is an $m_{A_J} \times 1$ vector of response variables, $\bm{x}_{iA_J} = (\bm{x}_{ij})_{j \in A_J}$ is a $p \times m_{A_J}$ matrix of independent variables, $\bm{f}_{A_J} = (\bm{f}_{j})_{j \in A_J}$ is a $r \times m_{A_J}$ matrix of latent factors, and $\bm{u}_{iA_J}$ is a $m_{A_J} \times 1$ error term. 

Let $\bm{D} = \bm{I}_r \otimes \bm{\iota}_{m_r}$, where $\bm{I}_r$ is the identity matrix of dimension $r$, $\bm{\iota}_{m_r}$ is a vector of ones of dimension $m_r = m_{A_J}/r$, and $\otimes$ denotes Kronecker product. We assume, for simplicity, that the number of groups per factor $m_r$ is an integer, because practitioners can always reorder elements after discarding those not in $A_{J} \cup B_{J}$. Let $\mathbb{M} = (\bm{D}' \bm{D})^{-1} \bm{D}'$ be a $r \times m_{A_J}$ matrix that creates $r$ averages of variables considering $m_r$ observations. Multiplying equation \eqref{panel:AT} by $\mathbb{M}$, we obtain the following $r$ equations: 
\begin{equation}
\overline{\bm{y}}_{i A_J} = \overline{\bm{x}}_{iA_J}'\bm{\beta}+\overline{\bm{f}}_{A_J}'\bm{\lambda}_{i}+\overline{\bm{u}}_{iA_J},  \label{average:AT}
\end{equation}
where, for instance, $\overline{\bm{y}}_{i A_J} = \mathbb{M} \bm{y}_{i A_J}$ denotes the vector of $r$ possible sample averages considering the elements of the vector $\bm{y}_{i A_J}$. Assuming that the $r \times r$ matrix $\overline{\bm{f}}_{A_J}'$ is invertible, we can solve for $\bm{\lambda}_i$:
\begin{equation}
  \bm{\lambda}_{i} =  \left[ \overline{\bm{f}}_{A_J}' \right]^{-1} \left( \overline{\bm{y}}_{i A_J} - \overline{\bm{x}}_{iA_J}'\bm{\beta} - \overline{\bm{u}}_{iA_J} \right).  \label{lambdai}
\end{equation}

Substituting equation \eqref{lambdai} into the augmented factor model \eqref{panel}, one obtains, for each $j \in A_0$, 
\begin{align}
  y_{ij} &= \bm{x}_{ij}' \bm{\beta} + \bm{f}_j' \left[ \overline{\bm{f}}_{A_J}' \right]^{-1} \left( \overline{\bm{y}}_{i A_J} - \overline{\bm{x}}_{iA_J}'\bm{\beta} - \overline{\bm{u}}_{iA_J} \right) + u_{ij} \nonumber \\ 
         &= \overline{\bm{y}}_{i A_J}' \bm{\theta}_{j} + \bm{x}_{ij}' \bm{\beta} - \bm{\theta}_{j}' \overline{\bm{x}}_{i A_J}' \bm{\beta}  
      + ( u_{ij} - \bm{\theta}_{j}' \overline{\bm{u}}_{iA_J}),
\end{align}
where $\bm{\theta}_{j} = \overline{\bm{f}}_{A_J}^{-1} \bm{f}_j$. We emphasize that the parameter depends on the normalization $A_J$ but we omit the dependence to keep the notation simple. By noting that $\bm{\theta}_{j}' \overline{\bm{x}}_{i A_J}' \bm{\beta} = \sum_{k=1}^r \overline{\bm{x}}_{i A_J, k}' \bm{\beta} \theta_{t,k}$, we can write,
\begin{equation}
  y_{ij} =  \overline{\bm{y}}_{i A_J}' \bm{\theta}_{j}  + \bm{x}_{ij}' \bm{\beta} + \overline{\bm{X}}_{i A_J}' \bm{\gamma}_j + v_{ij},  \label{eq:reduced_form}
\end{equation}
where $\overline{\bm{X}}_{i A_J}$ is a vector of $p \times r$ independent variables, the vector $\bm{\gamma}_j = ( \bm{\beta}' \theta_{j,1},..., \bm{\beta}' \theta_{j,r})'$, and $v_{ij} = u_{ij} - \bm{\theta}_{j}' \overline{\bm{u}}_{iA_J}$. Although $\bm{\theta}_j$ and $\bm{\beta}$ could be estimated by standard methods for linear models, the variable in the first term of equation \eqref{eq:reduced_form}, $\overline{\bm{y}}_{iA_J}$, is endogenous because it is correlated with $\overline{\bm{u}}_{iA_J}$, which appears as part of the error term. 

We propose to estimate equation \eqref{eq:reduced_form} using internal instruments $\overline{\bm{y}}_{iB_J}$, as well as an expanded set of instruments, as described below. The assumptions imposed below imply that $\overline{\bm{y}}_{iB_J}$ is a strong and valid instrument (see the discussion after the main result). An additional challenge is related to inference because we are interested in estimating simultaneously $\bm{f}_j$ for all $j \in A_{0}$. We proceed by stacking the reduced form equation \eqref{eq:reduced_form}. To handle the dependence across equations $j \in A_0$ for a given $i$ within the system, we use the asymptotic theory for clusters in \citet{bHansen2019}.

Recall that the number of equations $m_{A_0}$ is fixed, in the sense that it does not diverge if $J$ does. The system of $m_{A_0}$ equations can be written as:
\begin{align}
\bm{y}_{iA_0} &=  \left( \bm{I}_{m_{A_0}} \otimes \overline{\bm{y}}_{iA_J}' \right) \bm{\theta}_{A_{0}}
 + \bm{X}_{iA_{0}} \bm{\beta}
 + \left( \bm{I}_{m_{A_0}} \otimes \overline{\bm{X}}_{iA_J} ' \right) \gamma_{A_{0}}
 + \bm{v}_{iA_{0}} \\
 &=: \bm{M}_{iA_0} \bm{\delta} + \bm{v}_{iA_0}, \label{eq:reduced_system}
\end{align}
where $\bm{X}_{iA_{0}} = (\bm{x}_{i1}', \bm{x}_{i2}', ..., \bm{x}_{im_{A_0}}')'$ is a $m_{A_0} \times p$ matrix of exogenous variables, and $\bm{v}_{iA_{0}}$ is a $m_{A_0}$ dimensional vector with typical element $u_{ij} - \bm{\theta}_j' \overline{\bm{u}}_{iA_J}$. The parameter $\bm{\delta} = (\bm{\theta}_{A_{0}}', \bm{\beta}', \bm{\gamma}_{A_{0}}')'$, where $\bm{\theta}_{A_{0}} = (\bm{\theta}_{1}', \bm{\theta}_{2}', \hdots, \bm{\theta}_{m_{A_0}}')'$ and $\bm{\gamma}_{A_{0}} = (\bm{\gamma}_{1}', \bm{\gamma}_{2}', \hdots, \bm{\gamma}_{m_{A_0}}')'$. The total number of parameters in the system of equations \eqref{eq:reduced_system} is $k_{A_0} := m_{A_0} r (1+p)+p$. 

The Grouped Variable Estimator (GVE) can be obtained as:
\begin{footnotesize}
\begin{equation}
\widehat{\bm{\delta}}  =  \left(\sum_{i=1}^{N}\bm{M}_{iA_0}'\bm{Z}_{iA} \left(\sum_{i=1}^{N}\bm{Z}_{iA}'\bm{Z}_{iA}\right)^{-1}  \sum_{i=1}^{N}\bm{Z}_{iA}'\bm{M}_{iA_0}\right)^{-1}  \sum_{i=1}^{N}\bm{M}_{iA_0}'\bm{Z}_{iA}\left(\sum_{i=1}^{N}\bm{Z}_{iA}'\bm{Z}_{iA}\right)^{-1}\sum_{i=1}^{N}\bm{Z}_{iA}'\bm{y}_{iA_{0}},\label{eq:2sls_for_afm}
\end{equation}%
\end{footnotesize}%
where $\bm{Z}_{iA}$ denote a matrix of internally generated instruments. For instance, stacking the instrumental variables analogously, we obtain the instrumental variables
\begin{equation}
\bm{Z}^{(1)}_{iA}=\left[ \left( \bm{I}_{m_{A_0}} \otimes \overline{\bm{y}}_{iB_J}' \right) 
                          \quad 
                          \bm{X}_{iA_{0}}
                          \quad
                          \left( \bm{I}_{m_{A_0}} \otimes \overline{\bm{X}}_{iA_J}' \right)
                          \right]. \label{eq:instruments_1}
\end{equation}
where $\overline{\bm{y}}_{iB_J}$ is a $r$-dimensional vector of individual specific averages. The assumptions we maintain below actually imply a richer set of instruments, namely
\begin{equation}
\bm{Z}^{(2)}_{iA}=\left[ \left( \bm{I}_{m_{A_0}} \otimes \bm{y}_{iB_J}' \right)  \quad  \left(\bm{I}_{m_{A_0}} \otimes \bm{x}_{i}\right) \right]. \label{eq:instruments_2}
\end{equation}
The first set of instruments, $\bm{Z}^{(1)}_{i,A}$, leads to a just-identified IV estimator, regardless of (a potentially divergent) $m_{A_J}$. The second set of instruments $\bm{Z}^{(2)}_{i,A}$ is larger, and the number of elements will diverge if $m_{A_J}+m_{B_J}$ diverges.

\subsection{Identification}

In a factor model, $\bm{\lambda}_i$ and $\bm{f}_j$ are identified up to a non-singular transformation. To see this, note that the second term in equation \eqref{panel} corresponding to the partition $\bm{y}_{iA_0}$ can be written as $\bm{f}_{A_0}' \bm{A} \bm{A}^{-1} \bm{\lambda}_{i}$ for any non-singular $\bm{A}$ matrix of dimension $r \times r$. \citet{BAI2013} and \citet{bWilliams2020} discuss restrictions imposed to achieve point identification of factors and loadings. One set of restrictions on the $r^2$ free parameters is to normalize the upper $r \times r$ block of a matrix of loadings or factors. Thus, in the case that $m_{A_J}=r$, it is standard to consider $\bm{A} = \bm{f}_{A_J}^{-1}$, which has been used for identification using instrumental variables in \cite{HEATON2012348}, \citet{jHeckman1987}, and \citet{sPudney81}, among others. In these models, the first $r$ factors are normalized to one. 

\begin{example} \label{example0}
Consider equation \eqref{panel} with no regressors, $r=1$, $A_0 = 1$, $A_J = \{2,3\} $, and $B_J = \{4,5\}$. In this case, $m_r = 2$. The response variable in equation \eqref{average:AT} is $\bar{y}_{iA_J} = (y_{i2}+y_{i3})/2$ and $\bar{f}_{A_J} = (f_{2}+f_{3})/2$. The parameter $\theta_{A_0} = m_r \left( \sum_{j=1}^{m_r} f_{j+1}^{-1} f_1 \right)$, or simply $\theta_{A_0} = 2 f_1 / (f_2 + f_3)$. 
\end{example} 

Note that $\bm{\theta}_{j} = \overline{\bm{f}}_{A_J}^{-1} \bm{f}_j$ uses a different non-singular transformation than the one typically considered in the context of instrumental variables. Our transformation for the linear factor model leads to a normalization based on average of factors, which is convenient in terms of interpretation. After the factor model is estimated by \eqref{eq:2sls_for_afm}, we can employ transformations to uncover a simpler parameter structure. As an illustrative example, we can consider Example \ref{example0}. In this case,
\begin{equation*}
\frac{\theta_{A_0}}{\theta_{A_0} + m_r} = \frac{f_1}{f_1 + f_2 + f_3},
\end{equation*}  
showing that a simple reparametrization identifies the relative importance of the first factor. Naturally, there are other non-singular transformations that can be considered including $\bm{A} = \bm{f}_{A_J}^{-1}$, as discussed in the next section.

To think about identification of factors using instrumental variables, it is instructive to consider a special case when $m_r = 1$. 

\begin{example} \label{example1}
Suppose $y_{ij}$ in equation \eqref{panel} is the grade of student $i$ in mathematics, reading, and writing, and we are interested in estimating how teacher's quality affects academic performance (see Section \ref{Colombia}). For simplicity, consider a simple case with no regressors, $r=1$, $A_0 = 1$ and $A_J = 2 $, and $B_J = 3$. Then, $\theta_{A_0} = \theta_1 = f_1 / f_2$, and the estimator defined in \eqref{eq:2sls_for_afm} is:
\begin{equation}
\hat{\theta}_1=\frac{\frac{1}{N}\sum_{i=1}^{N} y_{i3} y_{i1}}{\frac{1}{N}\sum_{i=1}^{N} y_{i3} y_{i2}}. \label{fiv}
\end{equation}
\end{example}

\subsection{Large sample results}

In this section, we establish conditions under which the estimator in \eqref{eq:2sls_for_afm} is consistent and asymptotically normal. We will leverage the fact that our estimator can be viewed as a two stage least squares estimator for clustered data, where the cross-section units $i$ are the clusters; the measurements $j$ are observations within a cluster; the dependent variables are $\bm{y}_{iA_0}$; endogenous regressors are $\bm{y}_{iA_J}$; and so on. This allows us to use the asymptotic theory for clustered samples in \cite{bHansen2019}, in particular their results for two stage least squares estimation in Theorems 8 and 9.

To state our results, define
\begin{align*}
\bm{Q}_{N} & =\frac{1}{N\times m_{A_0}}\sum_{i=1}^{N}E\left[\bm{Z}_{iA}'\bm{M}_{iA_0}\right],\\
\bm{W}_{N} & =\frac{1}{N\times m_{A_0}}\sum_{i=1}^{N}E\left[\bm{Z}_{iA}'\bm{Z}_{iA}\right],\\
\bm{\Omega}_{N} & =\frac{1}{N\times m_{A_0}}\sum_{i=1}^{N}E\left[\bm{Z}_{iA}' \bm{v}_{iA_0} \bm{v}_{iA_0}' \bm{Z}_{iA}\right],\\
\bm{V}_{N} & =\left(\bm{Q}_{N} \bm{W}_{N}^{-1}\bm{Q}_{N}\right)^{-1}\bm{Q}_{N} \bm{W}_{N}^{-1}\bm{\Omega}_{N}\bm{W}_{N}^{-1}\bm{Q}_{N}\left(\bm{Q}_{N} \bm{W}_{N}^{-1} \bm{Q}_{N}\right)^{-1}.
\end{align*}

\begin{theorem}\label{th:HL_version}
Considering $\bm{Z}_{iA} = \bm{Z}^{(1)}_{iA}$ as in \eqref{eq:instruments_1}, if  

\noindent (a) $\left\{\left(y_{ij},\bm{x}_{ij}\right), j = 1,\cdots,J \right\}$ is independent across $i=1,\cdots,N$ conditional on $\bm{f}_j,j=1,\cdots,J$,
is generated by the factor model \eqref{panel},
and the choice of $A_{0},A_{J},B_{J}$ satisfies the restrictions outlined above,

\noindent (b) $E\left(\bm{\lambda}_{i}u_{ij}\right)=\bm{0}$ for all $i$ and for all
$j\in A_{J}$; $E\left(u_{ih}u_{ij}\right)=0$ for all $i$ and for
all $h\in A_{J}$ and $j \in B_{J}$; $E\left(\bm{x}_{ih}u_{ij}\right)=\bm{0}$
for all $i$ and for all $h,j \in A_{J}$, 

\noindent (c) for some $s>2$, $\sup_{i,j} E\left|y_{ij}\right|^{2s}<\infty$
and $\sup_{i,j,k}E\left|x_{ijk}\right|^{2s}<\infty$, 

\noindent (d) the matrix $\overline{\bm{f}}_{A_J}^{-1}$ is invertible,

\noindent and

\noindent (e) $\bm{Q}_{N}$ has full rank, $\lambda_{\text{min}}\left(\bm{\Omega}_{N}\right)\geq\lambda>0$,
and $\lambda_{\text{min}}\left(\bm{W}_{N}\right)\geq K>0$, where the smallest eigenvalue is denoted by $\lambda_{\text{min}}(\cdot)$.

\noindent Then, as $N \to \infty$, the estimator defined in \eqref{eq:2sls_for_afm} is consistent, $\widehat{\bm{\delta}}\stackrel{p}{\to}\bm{\delta}$, 
and
\begin{equation*}
\bm{V}_{N}^{-1/2}\sqrt{N \times m_{A_0}}\left(\widehat{\bm{\delta}}-\bm{\delta}\right)\stackrel{d}{\to}\mathcal{N}\left(0,\bm{I}\right).
\end{equation*}
\end{theorem}

\vspace{3mm}

The proof of Theorem \ref{th:HL_version} is presented in Appendix \ref{app:proofs}, and consists of verifying the conditions for Theorems 8 and 9 in \citet{bHansen2019}. Their requirement that the observations from each $i$ are asymptotically negligible (cf. their Assumption 1) for consistency is automatically satisfied, as our panel is balanced by assumption. Moreover, the condition stated in Assumption 2 in \citet{bHansen2019} requires that $m_{A_0} / N \to 0$, which is satisfied because $m_{A_0}$ does not grow with $J$. The asymptotic variance $\bm{V}_N$ can be consistently estimated in the usual way (see Theorem 9 in \citet{bHansen2019}).

The result in Theorem \ref{th:HL_version} is obtained considering several standard assumptions. Following Assumption (a), data are generated by model \eqref{panel}. Assumption (b) guarantees that the instruments are valid by requiring that the error terms in two partitions are not correlated. Assumption (c) is a boundedness condition on the regressors and outcome variable that allows for distributional heterogeneity, and is sufficient for \citet{bHansen2019}'s central limit theorem. Assumption (d) controls the behavior of the $\bm{f}_j$, which is part of the estimand and Assumption (e) asks for sufficient correlation of the instruments with the regressors.

\begin{example} \label{example1b}
It is straightforward to see that the estimator in \eqref{fiv} is consistent. First, under assumptions stated in Theorem \ref{th:HL_version}, note that
\begin{equation*}
\frac{1}{N}\sum_{i=1}^{N} y_{i3} y_{i1} = f_{3} f_{1}\frac{1}{N}\sum_{i=1}^{N} \lambda_{i}^{2}+ f_{3}\frac{1}{N}\sum_{i=1}^{N} \lambda_{i} u_{i1} + f_{1} \frac{1}{N}\sum_{i=1}^{N}  \lambda_{i} u_{i3} +\frac{1}{N}\sum_{i=1}^{N} u_{i1}u_{i3} \stackrel{p}{\rightarrow} f_{3} f_{1} E\left(\lambda_{i}^{2}\right).
\end{equation*}
A similar derivation can be employed for the numerator in equation \eqref{fiv} to find that $N^{-1} \sum_{i=1}^{N} y_{i3} y_{i2} \stackrel{p}{\to} f_{3} f_{2} E\left(\lambda_{i}^{2}\right)$. As a result, $\hat{\theta}_1 \stackrel{p}{\to} f_{1}/f_{2} =: \theta_1$.
\end{example}

The result of Theorem \ref{th:HL_version} holds under $J$ fixed or $J \to \infty$ because the number of parameters $k_{A_0}$ does not depend on $m_{A_J}$ or $m_{B_J}$,\footnote{Recall that we assume that $m_{A_0}$ does not diverge if $J$ does.} and the number of instruments $\bm{Z}^{(1)}_{iA}$ does not diverge with $J$. The estimator in Theorem \ref{th:HL_version} uses a fixed number of averages and not all the available internally generated instruments. If $J$ is fixed, it is straightforward to see that the result also holds for $\bm{Z}_{iA} = \bm{Z}^{(2)}_{iA}$ as in equation \eqref{eq:instruments_2}. The case of $\bm{Z}_{iA}^{(2)}$ and $J \to \infty$ is different, as the number of internal instruments increases with $J$, and therefore, the estimator \eqref{eq:2sls_for_afm} faces similar challenges to the ones found in the estimation of high-dimensional models \citep{svanderGeer2010, Belloni2012,fWind2019}. We investigate the large sample behavior of the estimator in Section \ref{sec:largeJ}.

The GVE estimator we propose have a number of attractive features. First, they are trivial to implement: they are 2SLS estimators with instruments \eqref{eq:instruments_1} or \eqref{eq:instruments_2} in the linear system \eqref{eq:reduced_system}. 
Second, the estimator with instruments $\bm{Z}_{iA}^{(1)}$ has the attractive property that it is consistent without any restrictions on the rate at which $J$ grows with $N$ while also being fixed-$J$ consistent.
Third, the estimator combine information from all units in $A_J$ and $B_J$, which are chosen by the researcher and are allowed to diverge.
Fourth, existing solutions for handling missing data in 2SLS settings can be used to handle unbalanced panels.
Fifth, we can easily accommodate regressors $\bm{x}_{ij}$ correlated with the error term $u_{ij}$ by using external instruments. Below we explore further the performance of our estimation approach in large $J$ settings and provide alternatives with improved finite sample performance.

A drawback of our approach is that it may not incorporate the information in the model efficiently. We offer the following two refinements, leaving careful study of their asymptotic properties for future research. 

First, note that our estimators are for the parameters 
$(\bm{\theta}_{jA_0}',\bm{\beta}',\bm{\gamma}_{jA_0}')$.
This overparametrization was chosen for ease of implementation. One could use efficient minimum distance to gain efficiency. Alternatively, we can consider a sequential approach based on a consistent estimator of $\bm{\beta}$ (see Section \ref{sec:multiple}).

Second, we could repeat the estimation result for different choices of $A_J$, as long as $A_0$, $A_J$, and $B_J$ satisfy the restrictions outlined in the text above. One would typically set $B_J = \{1,\cdots,J\} \setminus \left( A_0 \cup A_J \right)$. As long as the number of choices of $A_J$ does not diverge, the distributional results can be applied directly. 

\section{On adopting multiple normalizations}\label{sec:multiple}

The issue of multiple normalizations deserves further treatment as there are many situations where economic theory is silent on the type of restrictions imposed to the model. In those situations, practitioners face a possibly large number of normalizations that could be used to eliminate the problem of rotational indeterminacy of the factor model. Considering the first partition as the normalization might be arbitrary, as noted in a series of recent papers \citep*{oAttanasio2020,DelBono2020}. We briefly illustrate this issue using the following example:

\begin{example} \label{example2}
In example \ref{example1}, it is clear that $\hat{\theta}_{A_0}$ converges to $\theta_{A_0}=f_1/f_2$, which corresponds to \emph{a} normalization based on reading. The parameter $\theta_{A_0}$ can also be estimated by a normalization based on the factor for writing, $f_3$, and using $y_{i2}$ as an instrument for $y_{i3}$. 
\end{example}

Examples \ref{example1} and \ref{example2} illustrate that it not clear a priori whether to normalize based on mathematics, reading, or writing, leading to important practical questions on how to select a normalization and the corresponding partition. In fact, there are $Q_J$ ways of choosing the subset $A_J$, where
\begin{equation}
Q_J = {J - m_{A_0} \choose m_{A_J}} = \frac{(J - m_{A_0})!}{m_{A_J}! ~ ( J - m_{A_0} - m_{A_J})! } \leq 2^{J - m_{A_0}}. \label{nofn}
\end{equation} 
The solution we pursue in this section is to simultaneously adopt multiple subsets.

Theorem \ref{th:HL_version} establishes conditions under which the GVE that uses one normalization $A_J$ is consistent for the normalized factors and the regression coefficient. In this section, we will assume that $\bm{\beta}$ is known and focus on improved estimation of the factors by using information from multiple normalizations. Define $\bm{R}_i = ( \bm{R}_{iA_0}', \bm{R}_{iA_J}', \bm{R}_{iB_J}')'$ where, for instance, 
$$\bm{R}_{iA_0} = \bm{y}_{iA_0} - \bm{x}_{iA_0}'\bm{\beta} = \bm{f}_{A_0}' \bm{\lambda}_{i} + \bm{u}_{iA_0},$$ and $\bm{f}_{A_0}$ is a matrix of dimension $r \times m_{A_0}$. While the results of Theorem \ref{th:HL_version} hold for general $m_{A_J} \geq r$, we will focus on the special case $m_{A_J} = r$ \citep[e.g.,][among others]{madansky1964instrumental,sPudney81,jHeckman1987,HEATON2012348,bWilliams2020}. 

Letting $q = 1, 2, \hdots, Q_J$, for $Q_J$ as in \eqref{nofn}, and $\bm{M}_{iA_0,(q)} = \bm{I}_{m_{A_0}} \otimes \bm{R}_{iA_J,(q)}'$, we can write
\begin{equation}
\bm{R}_{iA_0} = \bm{M}_{iA_0,(q)} \bm{\theta}_{A_{0},(q)}  + \bm{u}_{iA_0} - \left( \bm{I}_{m_{A_0}} \otimes \bm{u}_{iA_J,(q)} ' \right) \bm{\theta}_{A_{0},(q)} \label{eq:sysres}
\end{equation}
where $\bm{\theta}_{A_0,(q)} = [\bm{I}_{m_{A_0}} \otimes \bm{f}_{A_J,(q)}^{-1}] \vect(\bm{f}_{A_0})$. Relative to the parameter $\bm{\theta}_{A_0}$ estimated by the GVE defined in \eqref{eq:2sls_for_afm}, the parameter $\bm{\theta}_{A_0,(q)}$ in equation \eqref{eq:sysres} can be based on any non-singular transformation. If $Q_J = 1$ and $m_r=1$, then  $\bm{\theta}_{A_0,(q)} = \bm{\theta}_{A_0}$. Moreover, we define
\begin{equation}
\bm{\vartheta}_{A_0} = \sum_{q=1}^{Q_J} W_{q}  \bm{\theta}_{A_0,(q)},  \label{eq:newNST}
\end{equation}
where $W_{q}$ is a matrix of weights. Below, we introduce conditions that cover weighting for models with $J$ fixed (Theorem \ref{thm:MGVE-fixed-T}) and $J$ increasing to infinity (Theorem \ref{T3}). They allow the use of different weighting schemes, including equal weighting.

Consider the following three illustrative examples.

\begin{example}
Consider Example \ref{example0}, where $A_0 = r = 1$ and $A_J = \{2,3\}$. Considering $W_q=1$, we have that $\vartheta_{A_0} = \vartheta_1 = f_1/f_2+f_1/f_3$ and therefore, $\vartheta_1 + 1 = \sum_{q=1}^3 f_q^{-1} f_1$. 
\end{example}

\begin{example}
Consider $A_0 = \{1,2\}$, $r=1$, and $A_J = \{3,4,5\}$. In this case, $Q_J=3$. With equal weights $W_q = 1/3$, we have that
\begin{equation*}
\bm{\vartheta}_{A_0} = \left(\begin{array}{c} \vartheta_1 \\ \vartheta_2 \end{array} \right) = \frac{1}{3} \sum_{q=1}^3 \left[ I_2 \otimes f_{A_J,(q)}^{-1} \right] \left(\begin{array}{c} f_1 \\ f_2 \end{array} \right) = \left(\frac{1}{3f_3} + \frac{1}{3f_4} + \frac{1}{3f_5} \right) \left(\begin{array}{c} f_1 \\ f_2 \end{array} \right).
\end{equation*}
and thus, $\bm{f} = (f_1,f_2)'$ is identified up to a non-singular transformation, which requires that factors in $A_J$ are bounded away from zero.  
\end{example}

\begin{example}
Consider now a two factor model with $J = 6$. In this case, there are $Q_J=6$ different ways of choosing the normalization. For simplicity, assume that $\bm{\vartheta}_{A_0}$ includes $\bm{\vartheta}_5$ and $\bm{\vartheta}_6$. Then,
\begin{equation*}
\bm{\vartheta}_{A_0} = \left(\begin{array}{c} \bm{\vartheta}_5 \\ \bm{\vartheta}_6 \end{array} \right) = \frac{1}{2} \left[ \bm{I}_2 \otimes \left[ \left(\begin{array}{cc} f_{1,1} & f_{1,2} \\ f_{2,1} & f_{2,2} \end{array} \right)^{-1} + \hdots + \left(\begin{array}{cc} f_{3,1} & f_{3,2} \\ f_{4,1} & f_{4,2} \end{array} \right)^{-1} \right] \right] \left(\begin{array}{c} \bm{f}_{5} \\ \bm{f}_{6} \end{array} \right).
\end{equation*}
Again, $\bm{f}_5 = (f_{5,1},f_{5,2})'$ and $\bm{f}_6 = (f_{6,1},f_{6,2})'$ are identified up to a non-singular transformation provided that the matrices corresponding to the first two partitions are non-singular. 
\end{example}

Then, we estimate \eqref{eq:newNST} by the weighted grouped variable estimator (WGVE):   
\begin{equation}
        \widehat{\bm{\vartheta}}_{A_0} = \sum_{q=1}^{Q_J} W_{q} \widehat{\bm{\theta}}_{A_0,(q)}, \label{MGVE}
\end{equation}
where $\widehat{\bm{\theta}}_{A_0,(q)}$ is the GVE defined in \eqref{eq:2sls_for_afm} based on partition $q$. The use of weights for the combination of estimators in a linear fashion is naturally not new \citep[e.g.,][among others]{hP06,linton2016,mHarding2020}. If $r=1$ then $Q_J = J - m_{A_0}$, and one could set $W_q = Q_J^{-1} = (J - m_{A_0})^{-1}$, and define the estimator as $\widehat{\bm{\vartheta}}_{A_0} = (J - m_{A_0})^{-1} \sum_{q=1}^{J - m_{A_0}} \widehat{\bm{\theta}}_{A_0,(q)}$, which is similar in spirit to the common correlated effect estimator of \citet*{hP06}. Moreover, the estimator \eqref{MGVE} is similar to the ones investigated by \citet*{linton2016}. For instance, Example 1 in \citet*{linton2016} consider a similar instrumental variable estimator for a simultaneous equation model, and the optimal choice of weights makes a weighted instrumental variable estimator asymptotically equivalent to the classical 2SLS estimator.

\subsection{Estimation in small J panels}

We begin by considering a panel data model when $J$ is fixed, and therefore, the number of non-singular transformations, $Q_J$, is constant. Note that $A_J$ and $B_J$ are fixed too, so the number of instruments employed in the first stage and the number of normalizations do not increase. This is the case most relevant in the applications using administrative data presented in Section \ref{empirical-application} and in the recent econometric literature \citep[see][for examples]{aJuodis2018,aJuodis2020,NORKUTE2021416}.

The estimator is a trivial extension of the method discussed in the previous section. In the first step, we obtain $\widehat{\bm{\theta}}_{A_{0},(q)}$ for $q = 1, 2, \hdots, Q_J,$ using the estimator \eqref{eq:2sls_for_afm}. In the second step, we compute a consistent estimator of $\bm{\vartheta}_{A_0}$ using a linear combination of consistent estimators obtained in the first step, as shown in \eqref{MGVE}. As expected, a linear combination of a finite number of consistent and asymptotically normal estimators is consistent and asymptotically normal, as shown in Theorem \ref{thm:MGVE-fixed-T} below. 


Let $\bm{v}_{iA_0,(q)} = \bm{y}_{iA_0} - \bm{M}_{iA_0,(q)} \bm{\theta}_{A_{0},(q)}$,  and consider the following definitions:
\begin{align*}
\bm{Q}_{N,(q)} & =\frac{1}{N\times m_{A_0}}\sum_{i=1}^{N}E\left[\bm{Z}_{iA}'\bm{M}_{iA_0,(q)}\right],\\
\bm{W}_{N,(q)} & =\frac{1}{N\times m_{A_0}}\sum_{i=1}^{N}E\left[\bm{Z}_{iA}'\bm{Z}_{iA}\right],\\
\bm{\Omega}_{N,(q)} & =\frac{1}{N\times m_{A_0}}\sum_{i=1}^{N}E\left[\bm{Z}_{iA}' \bm{v}_{iA_0,(q)} \bm{v}_{iA_0,(q)}' \bm{Z}_{iA}\right],\\
\bm{\Phi}_{N,(q)} & = \bm{Q}_{N,(q)} \bm{W}_{N,(q)}^{-1}\bm{Q}_{N,(q)}, \\
\bm{\Sigma}_{N,(q)} & = \bm{\Phi}_{N,(q)}^{-1} \bm{Q}_{N,(q)} \bm{W}_{N,(q)}^{-1}\bm{\Omega}_{N,(q)}\bm{W}_{N,(q)}^{-1}\bm{Q}_{N,(q)} \bm{\Phi}_{N,(q)}^{-1},
\end{align*}
and, by letting $\bm{\Sigma}_{(q)} := \lim_{N \to \infty} \bm{\Sigma}_{N,(q)}$, define 
\begin{equation*}
\mathcal{V}_{A_0} = \sum_{q=1}^{Q_J} \sum_{l=1}^{Q_J} \left[ W_{q} \bm{\Sigma}_{(q)}^{1/2} \right] \left[ W_{l} \bm{\Sigma}_{(l)}^{1/2}  \right]'.
\end{equation*}
The next result considers multiple normalizations and builds on Theorem \ref{th:HL_version}: 

\begin{theorem} \label{thm:MGVE-fixed-T} 

Under conditions (a), (b), and (c) of Theorem \ref{th:HL_version}, if 

\noindent (i) Condition (d) Theorem \ref{th:HL_version} holds for all $1 \leq q \leq Q_J$ such that $\bm{f}_{{A_J,(q)}}$ is a invertible partition matrix of dimension $r \times r$;

\noindent (ii) Condition (e) in Theorem \ref{th:HL_version} holds for all $1 \leq q \leq Q_J$ such that $\bm{Q}_{N,(q)}$ has full rank, $\lambda_{\text{min}}\left(\bm{\Omega}_{N,(q)}\right)\geq\lambda>0$, and $\lambda_{\text{min}}\left(\bm{W}_{N,(q)}\right)\geq K>0$;

\noindent (iii) The weights $\{ W_{q} \}_{q=1}^{Q_J}$ satisfy $\sum_{q = 1}^{Q_J} W_{q} = I$.

\noindent Then, as $N\to\infty$, the estimator $\widehat{\bm{\vartheta}}_{A_0}$ in \eqref{MGVE} is consistent and asymptotically normal with covariance matrix $\mathcal{V}_{A_0}$.
\end{theorem}

Assumptions (i) and (ii) are generalizations of Assumptions (d) and (e) in Theorem \ref{th:HL_version}. Condition (i) imposes restrictions to generate suitable non-singular transformations across all partitions, and condition (ii) guarantees a well-defined asymptotic distribution across feasible non-singular transformations. Lastly, condition (iii) allows the use of different weighting schemes to improve the performance of the GVE and it is similar to the ones employed in the literature such as \citet*{hP06} and \citet*{linton2016}. We do not consider random weights, but condition (iv) can be easily accommodated to incorporate random matrices as in \citet*{linton2016}.

Optimal weights can be found as minimizers of the asymptotic covariance matrix of the estimator. To see this, write $\mathcal{V}_{A_0}$ as  
\begin{equation*}
\mathcal{V}_{A_0} = \sum_{q=1}^{Q_J} \sum_{l=1}^{Q_J} W_{q} \bm{\Sigma}_{(lq)} W_{l}'.
\end{equation*}
Let $\bm{\Sigma} = [\bm{\Sigma}_{(lq)}]$ and $\iota_{Q_J}$ be a $Q_J$-dimensional vector of ones. Thus,
\begin{equation*}
W_{0q}^\ast = \left[ (\iota_{Q_J} \otimes I)' \bm{\Sigma}^{-1}  (\iota_{Q_J} \otimes I)\right]^{-1} \left[ (\iota_{Q_J} \otimes I)' \bm{\Sigma}^{-1} \right]_q. 
\end{equation*}
It follows that the estimator $\widehat{\bm{\vartheta}}_{A_0}^\ast = \sum_{q=1}^{Q_J} W_{0q}^\ast \widehat{\bm{\theta}}_{A_0,(q)}$ has asymptotic covariance matrix,
\begin{equation*}
\mathcal{V}_{A_0}^\ast = \left[ (\iota_{Q_J} \otimes I)' \bm{\Sigma}^{-1}  (\iota_{Q_J} \otimes I)\right]^{-1}.
\end{equation*}
In other words, the optimal weighting is proportional to the inverse of the asymptotic covariance matrix of the estimator. The estimation $\mathcal{V}_{A_0}^\ast$ is straightforward and follows the estimation of $\bm{\Sigma}$. See Section 6 in \citet*{linton2016} for specific details. 

\subsection{Estimation in large J panels} \label{sec:largeJ}

In the case of panel data models with large $J$, possibly larger than $N$, there are known issues with the approach above. As discussed before, there is a large number of possible non-singular transformations. Moreover, least squares estimation of the regression of endogenous variables on the instruments has poor finite sample properties, and therefore, the GVE estimator is expected to perform poorly in practice. The procedure could suffer from a finite sample bias problem similar to the one investigated in \citet*{Jerry2008} and \citet*{Chao2012}. Therefore, this section investigates the case of large $J$ considering developments in \citet*{Belloni2012} and \citet*{linton2016}, although the large $J$ situation is not common in the analysis of student administrative data (Section \ref{empirical-application}). The procedure in  \citet{Belloni2012} requires to approximate a large dimensional model by a low-dimensional sub-model. If some instruments are invalid, the procedure can be easily adapted to include the median estimator proposed by \citet*{fWind2019}.

We propose to estimate $\bm{\vartheta}_j$ for all $j \in A_0$ in two main steps, as before. We begin by describing the first step involving the use of instrumental variables. Let $\bm{R}_j = (R_{1j},R_{2j},\hdots,R_{Nj})'$ be an $N$-dimensional vector of dependent variables, $\bm{R}_{A_J,(q)}$ be an $N \times r$ matrix of endogeneous variables, and $\bm{R}_{B_J,(q)}$ be an $N $ by $m_{B_J} = J-m_{A_0}-r$ matrix of internal instruments. Let $\widehat{L}_{il,(q)} := \bm{R}_{i,B_J,(q)}' \widehat{\bm{\pi}}_{l,(q)}$ for $l = 1, 2, \hdots, r$, where $\widehat{\bm{\pi}}_{l,(q)}$ is a Lasso estimator defined as a solution of the following problem:
\begin{equation*}
\widehat{\bm{\pi}}_{l,(q)} = \argmin_{\bm{\pi}_{l,(q)} \in \bm{\Pi}_{(q)}} \sum_{i=1}^N ( R_{il,A_J,(q)} - \bm{R}_{i,B_J,(q)}' \bm{\pi}_{l,(q)} )^{2} + \frac{\lambda_l}{N} \|  \Upsilon_l  \bm{\pi}_{l,(q)} \|_1,  \label{method1}
\end{equation*}
where the parameter set $\bm{\Pi}_{(q)} \subseteq \RR^{(J-m_{A_0}-r)}$ and $\| \bm{b} \|_1$ is the standard $\ell_1$-norm defined as $\| \bm{b} \|_1 = \sum_i | b_i |$ for a generic constant $b_i$. The penalty loadings $\Upsilon_l$ and $\lambda_l$ are selected as in \citet*{Belloni2012}. We then collect the predictions $\hat{L}_{il,(q)}$ for all $1 \leq i \leq N$ and $1 \leq l \leq r$ to obtain a matrix $\widehat{\bm{L}}_{(q)}$ of dimension $N \times r$. In a second stage of the IV approach, we find $\widehat{\bm{\theta}}_{j,(q)}$ as the solution of the following equation:
\begin{equation}
G_{N}(\bm{\theta}_{j}) = \widehat{\bm{L}}_{(q)}' ( \bm{R}_{j} - \bm{R}_{A_J,(q)} \bm{\theta}_{j,(q)} ) = 0,
 \label{GVE3}
\end{equation}  
with $G_{(q)}(\bm{\theta}_{j}) = E(G_{N,(q)}(\bm{\theta}_{j}))$. 

The last step includes the following estimator: 
\begin{equation}
        \widehat{\bm{\vartheta}}_j = \sum_{q=1}^{Q_J^\ast} W_{q} \widehat{\bm{\theta}}_{j,(q)},  \label{MGVE-largeT}
\end{equation}
where the truncation parameter $Q_J^\ast < Q_J$ for all $J$, and $\widehat{\bm{\theta}}_{j,(q)}$ is a Lasso-type estimator obtained as the solution of \eqref{GVE3}. 

Before establishing large sample results when the number of normalizations tend to infinity as $N$ and $J$ tend to infinity, we emphasize two conditions that are standard in the literature. The linear model estimated in the first stage of the IV procedure uses Condition AS in \citet*{Belloni2012} to approximate a conditional expectation up to a small non-zero approximation error. Let $L_{il,(q)} := \bm{R}_{i,B_J,(q)}'  \bm{\pi}_{0l,(q)} + a_{il,(q)}$, $\max_{1 \leq l \leq r}  \| \bm{\pi}_{0l,(q)} \|_0 \leq s_{(q)} = o(N)$, and 
\begin{equation} 
\max_{1 \leq l \leq r}  \left[ \frac{1}{N} \sum_{i=1}^N E ( a_{il,(q)}^2 )  \right]^{1/2} \leq c_{s} = O_p\left( \sqrt{\frac{s_{(q)}}{N}} \right),  \label{AS}
\end{equation} 
for $l = 1, 2, \hdots r$. Condition AS in \citet*{Belloni2012}, reproduced in the previous equations in the context of a factor model, states that at most $s_{(q)}$ variables are needed to approximate well the conditional expectation $L_{il,(q)}$ with a small approximation error. This error is of the same order of magnitude than the estimator error, $\sqrt{s_{(q)}/N}$. Noting that this condition holds for a given $j,q$ pair, the assumption allows us to identify 
\begin{equation*}
\mathfrak{J}_{l,(q)} := \mbox{support}(\bm{\pi}_{0l,(q)}) = \left\{k \in \{1,2,\hdots, J - m_{A_0} - r \}: | \pi_{0lk,(q)} | > 0 \right\}.
\end{equation*}

Moreover, \citet*{bickel2009} and \citet*{Belloni2012} introduce conditions on the Gram matrix $\bm{M} = E ( \bm{R}_{iB_{J},(q)} \bm{R}_{iB_{J},(q)}' )$, which is not positive definite if $m_{B_J}  > N$. They propose a notion of ``restricted" positive definiteness for vectors in a restricted set. In our model, this set is defined as $\Delta_\mu = \{ \bm{\Psi} \in \RR^{J - m_{A_0} - r} : \| \bm{\Psi}_{\mathfrak{J}_{l,(q)}^c} \|_1 \leq \mu \| \bm{\Psi}_{\mathfrak{J}_{l,(q)}} \|_1, \bm{\Psi} \neq 0 \}$. We can then define
\begin{equation*}
\kappa_\mu^2( \bm{M} ) := \min_{\Psi \in \Delta, |\mathfrak{J}_{l,(q)}| \leq s_{(q)}}  s_{(q)} \frac{ \bm{\Psi}' \bm{M} \bm{\Psi} } { \| \bm{\Psi}_{\mathfrak{J}_{l,(q)}} \|_1^2 }.
\end{equation*}

We now establish the consistency and asymptotic normality of the estimator introduced in \eqref{MGVE-largeT}.  

\begin{theorem}\label{T3}
Consider: 

\noindent (a) Let $\mathcal{Q}_J = \{1,2,\hdots,Q_J\}$ and $\mathcal{Q}_J^\ast = \{1,2,\hdots,Q_J^\ast\}$ with $Q_J^\ast < Q_J$ for all $J$. The sequence of weights $\{ W_q \}$ satisfy condition (iii) in Theorem \ref{thm:MGVE-fixed-T} and, as $Q_J \to \infty$, with 
\begin{equation*}
\sup_{J \geq 1} \sum_{q=Q_J^\ast+1}^{Q_J} \| W_{q} \| \to 0,
\end{equation*}
where, for any integer $C \geq 1$,
\begin{equation*}
Q_J^\ast =  \frac{C r!}{(J - m_{A_0})^{r}} Q_J;
\end{equation*}

\noindent (b) For any $\mu>0$, a constant $k$ exists such that $\kappa_\mu(\bm{M}) \geq k > 0$ with probability tending to one as $N \to \infty$;

\noindent (c) Let $\tilde{R}_{il,(q)} = R_{il,(q)} - N^{-1} \sum_{i=1}^N E( R_{il,(q)} )$ and $e_{il,(q)} := R_{il,A_J,(q)} - \bm{R}_{i,B_J,(q)}' \bm{\pi}_{l,(q)}$. The variables $\tilde{R}_{il,(q)}$ and $e_{il,(q)}$ have uniformly bounded conditional moments of order 4 and are i.i.d. across $i$. The vector $\bm{R}_{i,B_J,(q)}$ is bounded and i.i.d. across $i$. Moreover, 
\begin{equation*}
\frac{K_N^2 (\log(m_{B_J} \vee N))^3}{N} \to 0, \; \mbox{and,} \; \; \frac{s_{(q)} \log(m_{B_J} \vee N) }{ N} \to 0;
\end{equation*}

\noindent (d) There is an $\epsilon_N(\eta) > 0$ for $\eta > 0$ with $\epsilon_N(\eta) \to 0$ when $N \to \infty$ such that
\begin{equation*}
\min_{q \in \mathcal{Q}_J^\ast} \inf_{ \| \bm{\theta} - \bm{\theta}_{0} \| > \eta} \| G_{(q)}(\bm{\theta}_j) \| \geq \epsilon_N(\eta) > 0;
\end{equation*}

\noindent and

\noindent (e) For $\epsilon_N(\eta)$, $Q_J^\ast$ and $Q_J$, with $Q_J^\ast \to \infty$, $Q_J \to \infty$ and $Q_J^\ast < Q_J$ as $J \to \infty$, there is a positive sequence $\alpha_{N} = o(1)$ with $\sup_N (\alpha_{N} / \epsilon_N(\eta)) < \infty$ such that
\begin{equation*}
\max_{q \in \mathcal{Q}_J^\ast} \sup_{ \bm{\theta}_j \in \bm{\Theta}_j}  \| G_{N,(q)}(\bm{\theta}_{j}) - G_{(q)}(\bm{\theta}_j) \|  = o_p(\alpha_{N}).
\end{equation*}

\noindent Under conditions (a)-(e) and the conditions of Theorem \ref{th:HL_version}, as both $N$ and $J \to \infty$, the estimator $\widehat{\bm{\vartheta}}_{j}$ defined in \eqref{MGVE-largeT} is consistent and asymptotically normal with covariance matrix $\bm{\Omega}_j$.
\end{theorem}

Condition (a) extends the condition used in Theorem \ref{thm:MGVE-fixed-T} to allow the use of weights when there is available a growing number of non-singular transformations, and it is similar to Condition A1 in \citet*{linton2016}. The implication is that truncation parameter $Q_J^\ast$ needs to grow slowly to satisfy the conditions in Lemma 1 in \citet*{linton2016} and are satisfied if $Q_J^\ast$ grows at logarithm rates. The truncation parameter defined in (a) satisfy the condition. Under assumption (b), we can determine the rate of convergence of the Lasso-type estimator in the case of Gaussian models with homocedastic errors. Assumption (c) is needed for the estimation of conditional expectation functions under non Gaussian conditions and heteroskedastic errors, and it is similar to Condition RF as implied by Lemma 3.b in \citet*{Belloni2012}. Conditions (b) and (c), in addition to the conditions on sparsity, are crucial for the consistency of the IV estimator. Assumption (d) is a modified version of a standard condition for uniform convergence of estimators that minimize a criterion function \citep[Theorem 5.9,][]{vanderVaart98}. The difference is that the condition is imposed on every normalization. Assumption (e) is similar to condition (A.4) in \citet*{linton2016}. These assumptions impose uniformity over $Q_J^\ast$ normalizations. See Lemma 1 in \citet*{linton2016}.

The result in Theorem \ref{T3} is achieved by using a ``standardized'' binomial coefficient as a truncation parameter, which controls the rate of growth of $Q_J^\ast$ as $J \to \infty$. As long as $r$ is fixed, one can approximate $Q_J$ by $(J-m_{A_J})^r/r!$, and thus, the ratio of $Q_J^\ast \to C$ as $J \to \infty$. This result is important in practice as it indicates that the truncation parameter $Q_J^\ast$ should be determined to minimize computational time as well as to maximize efficiency gains. There are several options for practitioners. One is to set $C=1$ and then potentially investigate the marginal impact of an additional normalization in terms of the standard error of the estimator. 
 
\section{Simulation experiments}

In this section, we conduct an investigation of the performance of the proposed approaches in comparison to existing methods. Using a series of simulation experiments, we report the root mean squared errors of new and existing estimators for different models. We first consider a factor model, and then a factor-augmented model.

We begin by considering a one-factor model similar to the one used in \citet{HEATON2012348}. Observations are generated from $y_{ij} = \lambda_{1,i} f_{1,j} + u_{ij}$, where the error term $u_{ij} \sim F_u$, and $\lambda_{1,i}$ is drawn as an independent observation from a uniform distribution ranging from 0.5 to 3.5. We generate observations for the factor following the equation $f_{1,j} = 0.8 f_{1,j-1} + \eta_{j}$ for $j = -S + 1, -S + 2, \hdots, 0, 1, \hdots, J$, where $\eta_{j}$ is an i.i.d. random variable distributed as uniform with support ranging from 0 to 1. We set $S = 50$ to minimize the effects of the initial value, $f_{1,-49}=1$. We consider two variations of the model. We first assume that the error term $u_{ij}$ is an i.i.d. Gaussian random variable, and then we assume that $u_{ij}$ is a random variable distributed as $t$-student with 3 degrees of freedom ($t_3$).

Table \ref{table0} presents the root mean squared error (RMSE) for the parameter $\bm{\theta}$. The table shows results from different estimators. We compare our estimators to more traditional approaches such as PCA \citep{BAI2013}, IV an instrumental variables estimator that uses internal instrumental variables, and LAS an estimator that uses the LASSO instead of the IV estimator. The implementation of the LASSO estimator utilizes the R package hdm as described by \cite{chernozhukov2016high}. For these estimators, we consider the root mean squared deviations of the estimated normalized factors, $(J^{-1} \sum_{j=1}^J (\hat{\theta}_{j} - \theta_{j})^2)^{1/2}$, where $\theta_j = f_{1,j}/f_{1,1}$. The table also shows the RMSE of the new estimators. GVE denotes the grouped variable estimator as in \eqref{eq:2sls_for_afm} using averages of $J/2$ instruments. For the GVE, we define $\theta_j = f_{1,j}/\bar{f}_{1}$, where $\bar{f}_{1}$ is constructed as an average of $J/2 - 1$ factors. Finally, WGVE refers to the weighted estimator that uses all partitions and a LASSO procedure in the first step as in \eqref{MGVE-largeT}. The RMSE is defined as $(J^{-1} \sum_{j=1}^J ( \hat{\vartheta}_{j} - \vartheta_{j} )^2)^{1/2}$, considering $Q_J^\ast = C = Q_J$. The table shows results for different sample sizes of $N = \{50,100\}$ and $J = \{10,20\}$. 

\begin{table}
\begin{center}
\begin{small}
\begin{tabular}{ c c c c c c c c c c c c} \hline
\multicolumn{2}{c}{}&\multicolumn{5}{c}{Model with Gaussian Errors} &\multicolumn{5}{c}{Model with $t_3$ Errors} \\ 
\multicolumn{1}{c}{N}&\multicolumn{1}{c}{J}&\multicolumn{1}{c}{PCA}&\multicolumn{1}{c}{IV}&\multicolumn{1}{c}{LAS}&\multicolumn{1}{c}{GVE}&\multicolumn{1}{c}{WGVE} &\multicolumn{1}{c}{PCA}&\multicolumn{1}{c}{IV}&\multicolumn{1}{c}{LAS}&\multicolumn{1}{c}{GVE}&\multicolumn{1}{c}{WGVE}  \\  \hline
50      &       10      &       0.034   &       0.034   &       0.034   & 0.030 & 0.026   &       0.060   &       0.059   &       0.058   &       0.053 & 0.043 \\
50      &       20      &       0.040   &       0.040   &       0.039   &       0.028 & 0.027 &       0.068   &       0.073   &       0.066   &       0.049 & 0.046   \\
100     &       10      &       0.026   &       0.026   &       0.026   &       0.022 & 0.018 &       0.046   &       0.042   &       0.042   &       0.044 & 0.031   \\
100     &       20      &       0.026   &       0.026   &       0.026   &       0.020 & 0.019 &       0.047   &       0.047   &       0.044   &       0.036 & 0.033   \\   \hline
\end{tabular}
\vspace{3mm}
\caption{\emph{Root mean squared error of estimators in a one-factor model. PCA refers to principal components analysis, IV denotes the instrumental variable estimator, LAS denotes the LASSO estimator, GVE is the grouped variable estimator, and WGVE is the weighted grouped variable estimator.} \label{table0}}
\end{small}
\end{center}
\end{table}

As can be seen from Table \ref{table0}, the PCA estimator offers excellent performance among existing methods. The IV estimator only performs marginally better than PCA in models with $t_3$ errors, although the difference in performance disappears when $N=100$ and $J=20$. Furthermore, it is interesting to see, although expected, that LASSO outperforms IV when $J$ is relatively large with respect to $N$. The results reflect the well known issues with IV estimation in factor models, while simultaneously demonstrating the advantages of employing the LASSO regression approach for high-dimensional models. In contrast the performance of the proposed GVEs is excellent and, in general, they offer the smallest RMSE across all variants of the model. 

We also investigate the relative performance of the estimator in a  factor-augmented panel data model. Following closely \citet{hP06}, we generate observations based on the following model for $i = 1, 2, 3, \hdots, N$ and $j = -S + 1, \hdots, 0, 1, \hdots, J$:
\begin{eqnarray}
y_{ij} & = & \beta_{0} + \beta_1 x_{1,ij} + \beta_{2} x_{2,ij} + \lambda_{1,i} f_{1,j} + u_{ij}, \label{mc:eqy} \\
x_{s,ij} & = &  a_j \lambda_{1,i} + b_j  f_{1,j} + c_j \lambda_{1,i} f_{1,j} + v_{s,ij},  \label{mc:eqx} \\
f_{1,j} & = &  \rho f_{1,j-1} + \eta_{j}. 
\end{eqnarray}
As in the case of the previous one factor model, the error term in equation \eqref{mc:eqy} is assumed to be distributed as either Gaussian or $t_3$. The error term in equation \eqref{mc:eqx} is $\bm{v}_{ij} = (v_{1,ij},v_{2,ij})' \sim \mathcal{N}(0,\bm{I})$. Moreover, we set the parameters of the model to generate an endogenous variable, $x_1$, and an exogenous variable, $x_2$. The parameters in equations \eqref{mc:eqy} and \eqref{mc:eqx} are $\beta_1=\beta_2=a_1=1$, $b_1=2$, $c_1=0.5$, $\beta_0=a_2=b_2=c_2=0$, and $\rho=0.8$. Lastly, as before, we set $S = 50$ to minimize the effects of the initial values on the outcome, $f_{1,-49}=1$, and $\eta_j$ is an i.i.d. random variable distributed as uniform $\mathcal{U}[0,1]$. 

\begin{table}
\begin{center}
\begin{small}
\begin{tabular}{ c c c c c c c c c c c c} \hline
\multicolumn{2}{c}{} &\multicolumn{5}{c}{Model with Gaussian Errors} &\multicolumn{5}{c}{Model with $t_3$ Errors} \\ 
\multicolumn{1}{c}{N}&\multicolumn{1}{c}{J}&\multicolumn{1}{c}{PCA}&\multicolumn{1}{c}{IV}&\multicolumn{1}{c}{LAS}&\multicolumn{1}{c}{GVE}&\multicolumn{1}{c}{WGVE}&\multicolumn{1}{c}{PCA}&\multicolumn{1}{c}{IV}&\multicolumn{1}{c}{LAS}&\multicolumn{1}{c}{GVE}&\multicolumn{1}{c}{WGVE} \\ \hline
 \multicolumn{2}{c}{}& \multicolumn{10}{c}{First Stage Method: IEE; Design 1: $\lambda_{1,i} \sim \mathcal{U}[0.5,3.5]$.} \\  \hline
50      &       10      &       0.050   &       0.049   &       0.050   &       0.040 & 0.036 &       1.417   &       0.106   &       0.366   & 0.110 &       0.157   \\
50      &       20      &       0.043   &       0.044   &       0.043   &       0.033 & 0.033 &       0.366   &       0.100   &       0.141   & 0.074 &       0.223   \\
100     &       10      &       0.030   &       0.030   &       0.030   &       0.027 & 0.027 &       0.328   &       0.076   &       0.093   & 0.065 &       0.067   \\
100     &       20      &       0.032   &       0.031   &       0.032   &       0.023 & 0.023 &       0.638   &       0.070   &       0.086   & 0.050 &       0.066   \\ \hline
 \multicolumn{2}{c}{}& \multicolumn{10}{c}{First Stage Method: CCE; Design 1: $\lambda_{1,i} \sim \mathcal{U}[0.5,3.5]$.} \\  \hline
50      &       10      &       0.045   &       0.044   &       0.044   &       0.036 & 0.034 &       0.075   &       0.071   &       0.072   & 0.063 &       0.058   \\
50      &       20      &       0.038   &       0.041   &       0.039   &       0.029 & 0.030 &       0.067   &       0.071   &       0.066   & 0.050 &       0.049   \\
100     &       10      &       0.033   &       0.033   &       0.033   &       0.027 & 0.027 &       0.060   &       0.054   &       0.056   & 0.045 &       0.041   \\
100     &       20      &       0.028   &       0.028   &       0.028   &       0.021 & 0.022 &       0.046   &       0.048   &       0.045   & 0.035 &       0.036   \\  \hline
 \multicolumn{2}{c}{}& \multicolumn{10}{c}{First Stage Method: IEE; Design 2: $\lambda_{1,i} \sim \mathcal{N}(0,1)$.} \\    \hline
50      &       10      &       0.078   &       0.080   &       0.080   & 0.067   & 0.062 &       0.453   &       0.150   &       0.211 & 0.138 & 0.433   \\
50      &       20      &       0.083   &       0.096   &       0.084   &       0.063 & 0.066 &       2.745   &       0.183   &       0.148   & 0.119 &       0.125   \\
100     &       10      &       0.058   &       0.058   &       0.058   & 0.048   & 0.044 &       0.204   &       0.102   &       0.109   & 0.088 &       0.082   \\
100     &       20      &       0.060   &       0.064   &       0.061   & 0.044   & 0.047 &       0.434   &       0.141   &       0.122   & 0.093 &       0.107   \\ \hline
 \multicolumn{2}{c}{}& \multicolumn{10}{c}{First Stage Method: CCE; Design 2: $\lambda_{1,i} \sim \mathcal{N}(0,1)$.} \\    \hline
50      &       10      &       0.075   &       0.078   &       0.077   &       0.065 & 0.062 &       0.159   &       0.132   &       0.126 & 0.110 & 0.114   \\
50      &       20      &       0.082   &       0.096   &       0.083   &       0.062 & 0.066 &       0.151   &       0.182   &       0.141   & 0.108 &       0.117   \\
100     &       10      &       0.057   &       0.057   &       0.057   & 0.047   & 0.044 &       0.103   &       0.098   &       0.096   & 0.080 &       0.076   \\
100     &       20      &       0.060   &       0.064   &       0.061   & 0.044   & 0.047 &       0.121   &       0.127   &       0.107   & 0.075 &       0.083   \\ \hline
 \multicolumn{2}{c}{} &  \multicolumn{10}{c}{First Stage Method: IEE; Design 3: $\lambda_{1,i}$ distributed as $\mathcal{N}$ or $\mathcal{U}$.} \\    \hline
50      &       10      &       0.084   &       0.086   &       0.086   & 0.072 & 0.066   &       0.767   &       0.167   &       0.204   & 0.217 &       0.148   \\
50      &       20      &       0.091   &       0.105   &       0.090   & 0.067 & 0.071   &       0.392   &       0.205   &       0.177   & 0.120 &       0.162   \\
100     &       10      &       0.059   &       0.060   &       0.060   & 0.049 & 0.047   &       0.233   &       0.107   &       0.106   & 0.092 &       0.083   \\
100     &       20      &       0.061   &       0.066   &       0.062   & 0.046 & 0.050   &       0.524   &       0.133   &       0.133   & 0.086 &       0.102   \\ \hline
 \multicolumn{2}{c}{} &  \multicolumn{10}{c}{First Stage Method: CCE; Design 3: $\lambda_{1,i}$ distributed as $\mathcal{N}$ or $\mathcal{U}$.} \\    \hline
50      &       10      &       0.083   &       0.085   &       0.085   & 0.069 & 0.065   &       0.266   &       0.153   &       0.231   & 0.123 &       0.128   \\
50      &       20      &       0.089   &       0.104   &       0.090   & 0.067 & 0.071   &       0.191   &       0.203   &       0.149   & 0.110 &       0.149   \\
100     &       10      &       0.058   &       0.058   &       0.058   & 0.048 & 0.046   &       0.115   &       0.101   &       0.099   & 0.085 &       0.080   \\
100     &       20      &       0.061   &       0.066   &       0.062   & 0.046 & 0.050   &       0.256   &       0.128   &       0.110   & 0.080 &       0.089   \\ \hline
\end{tabular}
\vspace{3mm}
\caption{\emph{Root mean squared error of estimators in a factor-augmented panel data model. PCA refers to principal components analysis, IV denotes instrumental variable estimator, LAS denotes the LASSO estimator, GVE is the grouped variable estimator, and WGVE is the weighted grouped variable estimator. IEE refers to the interactive effects estimator and CCE refers to the comon correlated effects estimator.} \label{t2:mc}}
\end{small}
\end{center}
\end{table}

The focus of this investigation is on the estimation of the latent factor structure in the model, therefore we implement our procedure by first estimating the intercept and slopes of the observed part of equation \eqref{mc:eqy}. We employ two consistent estimators: the estimator for an interactive effects model (IEE) proposed by \citet{jB09}, and the mean group estimator for the common correlated effects model (CCE) developed by \citet{hP06}. Moreover, we evaluate the performance of the method in relation to the correlation between endogeneous variables and instruments. In Design 1, we assume that $\lambda_{1,i}$ is an i.i.d. random variable distributed as $\mathcal{U}[0.5,3.5]$, while in Design 2, $\lambda_{1,i}$ is an i.i.d. random variable distributed as $\mathcal{N}(0,1)$. Lastly, in Design 3, we generate $\lambda_{1,i} \sim \mathcal{N}(0,1) $, for $i=1,...,m$ and, following \citet{aC2011}, $\lambda_{1,i} = 0.5 \times \theta_i / \sum_i \theta_i$ for $i=m+1,...,N$, where $\theta_i \sim \mathcal{U}[0,1]$ and $m = 0.9 N$. 

We evaluate the performance of our proposed estimators for a factor-augmented model against standard approaches involving PCA and IV. As previously discussed, in the context of a factor-augmented linear model, we can conceive of the feasible estimation of the factor structure in two steps. First, a consistent estimate of the coefficients on the observed variables is necessary to generate residuals.  Second, we apply either PCA, IV, LASSO, or the grouped variable estimators proposed in this paper (GVE and WGVE) to estimate the latent factor structure. In Table \ref{t2:mc}, we present the small sample performance of different estimators measured by the RMSE, which is defined as in Table \ref{table0}. The results  show that the performance of the estimators proposed in this paper leads to significant improvements in terms of RMSE relative to the PCA and IV-type approaches.

\section{Applications using administrative student data} \label{empirical-application}

In order to demonstrate the usefulness of our methods, we now present three applications to models of educational attainment. Depending on the data and setting, the latent factors estimated by our methods can be interpreted as measures of unobserved teaching quality. Furthermore, we can quantify the distribution of latent ability of the students.  First, we investigate how the distribution of latent abilities changes over subsequent years of K12 education. Second, we investigate how the distribution of latent abilities changes over subsequent years of college education, and lastly we investigate the change of the distribution of student ability after the implementation of a voucher program designed to improve educational outcomes. While these examples rely on different data sets and come from different countries, they highlight the usefulness of our techniques in varied settings in order to quantify the unobserved dimensions of student and teacher quality.

\subsection{Educational opportunity in the US}

In our first example, we present evidence on the temporal and geographic variability of educational opportunity across the US by using administrative data from over 11,000 school districts \citep{reardon2019educational}. We can model the district-level test scores using the following model, which accounts for the impact of latent school-district and grade-level heterogeneity on educational attainment using a one-factor specification: 
\begin{equation}
y_{ij} = \mu_{i} + \mu_{j} + \lambda_{i} f_{j} + u_{ij}. \label{educ:eq0} 
\end{equation}
Here $y_{ij}$ is the average normalized test score in district $i$ in grade $j$. The model also includes district fixed effects, $\mu_{i}$, and grade fixed effects, $\mu_{j}$. In this model, $\lambda_{i}$ is associated with district educational attainment, and the factor $f_{j}$ is interpreted as measuring educational attainment by grade $j$. Moreover, the term $\lambda_i f_{j}$ represents the interaction between educational attainment in district $i$ and quality of instruction in grade $j$. The value of including these latent terms becomes evident once we consider that high grade teaching quality can have a modest effect on the educational attainment of relatively under-performing districts, although it can dramatically impact performance in over-performing districts. 

\begin{table}[hptb]
\begin{center}
\begin{small}
\begin{tabular}{c c c c c c c c c c c c c} \hline 
\multicolumn{1}{c}{Grade}& \multicolumn{6}{c}{Mathematics}&\multicolumn{6}{c}{Reading}\\
& \multicolumn{2}{c}{PCA}&\multicolumn{2}{c}{IV}& \multicolumn{2}{c}{WGVE}&\multicolumn{2}{c}{PCA}&\multicolumn{2}{c}{IV}& \multicolumn{2}{c}{WGVE} \\
   & $\hat{\theta}$ & $\Delta$\%  &  $\hat{\theta}$ & $\Delta$\% &  $\hat{\vartheta}$ & $\Delta$\% &  $\hat{\theta}$ & $\Delta$\% &  $\hat{\theta}$ & $\Delta$\% &  $\hat{\vartheta}$ & $\Delta$\% \\ \hline
3       &       1.000   &               &       1.000   &               &       0.930   &               &       1.000   &               &       1.000   &               &       0.953   &               \\
4       &       1.064   &       6.4     &       1.078   &       7.9     &       1.002   &       7.8     &       1.019   &       1.9     &       1.024   &       2.4     &       0.981   &       2.9     \\
5       &       1.056   &       -0.8    &       1.032   &       -4.3    &       0.990   &       -1.2    &       1.051   &       3.1     &       1.032   &       0.8     &       1.010   &       3.0     \\
6       &       1.059   &       0.3     &       0.997   &       -3.4    &       0.995   &       0.5     &       1.038   &       -1.2    &       0.974   &       -5.6    &       0.989   &       -2.1    \\
7       &       1.021   &       -3.6    &       0.949   &       -4.9    &       0.973   &       -2.2    &       1.008   &       -3.0    &       0.933   &       -4.2    &       0.967   &       -2.2    \\
8       &       0.984   &       -3.7    &       0.896   &       -5.6    &       0.924   &       -5.1    &       0.986   &       -2.1    &       0.900   &       -3.5    &       0.940   &       -2.8    \\
 \hline
\end{tabular}
\end{small}
\vspace{3mm}
\caption{\emph{Estimated factors by grade using the SEDA data. Principal Component Analysis is denoted by PCA, the instrumental variable estimator by IV, and the weighted grouped variable estimator by WGVE. $\Delta$\% denotes percentage change by grade} \label{table:us:factors}} 
\end{center}
\end{table}

To estimate the parameters in equation \eqref{educ:eq0}, we use data from the Stanford Education Data Archive (SEDA) for the year 2018 \citep{fahle2021stanford}. SEDA provides nationally comparable scores for school districts in the U.S. The data set includes information on mathematics and reading tests. Unfortunately, the availability of covariates that vary by grade and district is rather limited and it would reduce the sample size significantly. Thus, we do not include independent variables in the specification but account for district and grade fixed effects which we think will capture most of the relevant variation over a short time horizon.

We first estimate equation \eqref{educ:eq0} using fixed effects methods separately by subject. In the second stage, we use residuals $R_{ij} = \lambda_{i} f_{j} + u_{ij}$, and we estimate the factors and loadings following the the weighted grouped variable estimator (WGVE) with weights equal to $Q_J^{-1}$. In Table \ref{table:us:factors} we report the latent educational attainment for each grade (using grade 3 to normalize the results). We notice that for both mathematics and reading the three different estimators suggest a decreasing trend for higher grades. 

\begin{figure}
\begin{center}
\centerline{\includegraphics[width=.9\textwidth]{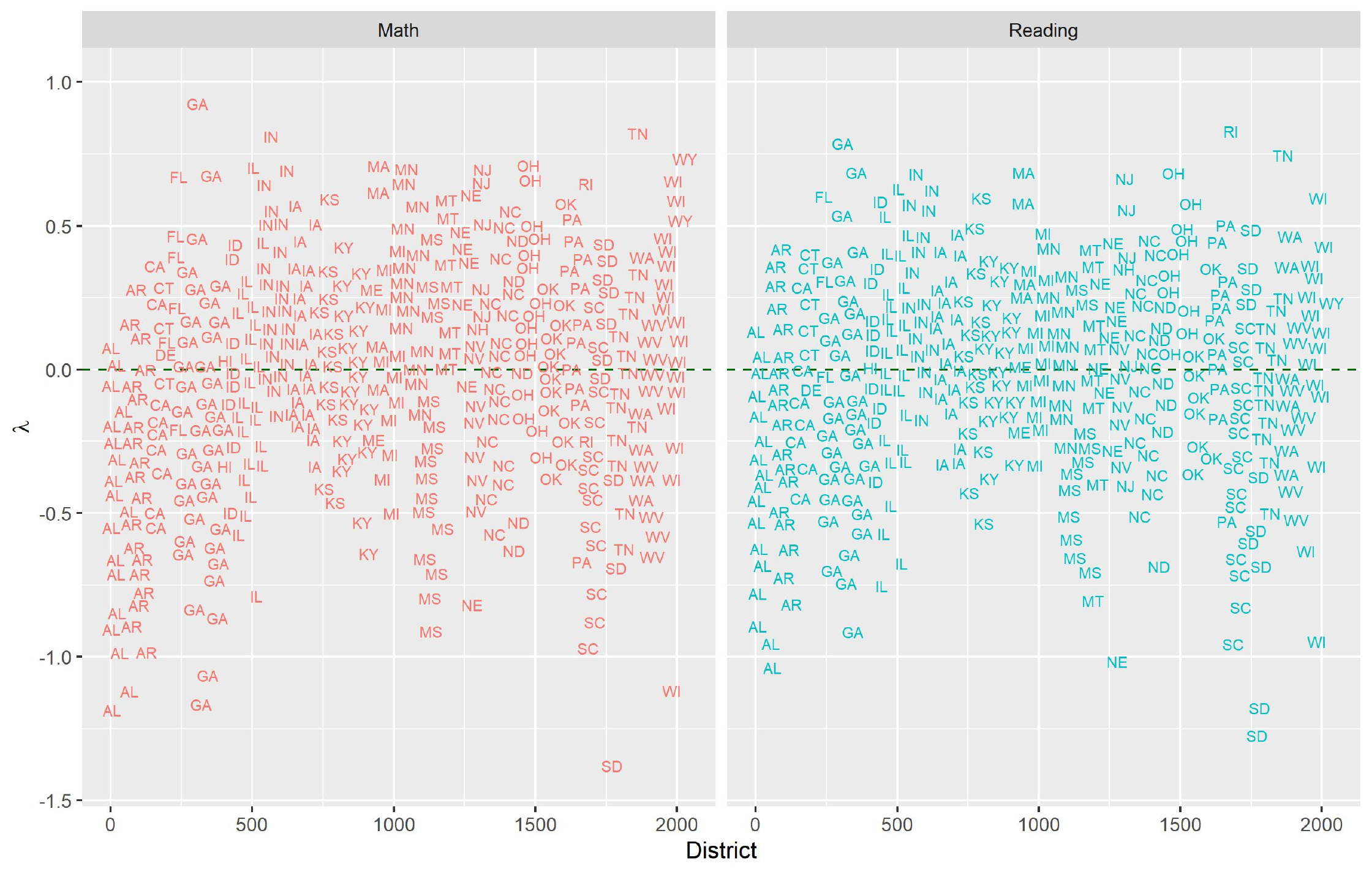}}
\caption{\emph{Geographical disparities and district educational attainment.} \label{f43:results}}
\end{center}
\end{figure}

In Figure \ref{f43:results} we plot the estimated model $\lambda$s for each school district grouped by state. The two counties with the lowest loadings for mathematics are Oglala Lakota County, SD (-1.379) and Todd County, SD (-1.367). They are the poorest and third poorest counties in the US respectively. In contrast the districts with the highest loadings are Forsyth County, GA (0.925) and Williamson County, TN (0.8218). Adjusted for cost of living they are some of the wealthiest counties in the US.  We also notice the distribution of estimated loadings by state. For Alabama the estimate loadings range from -1.209 to barely above zero at 0.074. In contrast the loadings for Massachusetts range from -0.207 to 0.590. It is worth noting that while Georgia has the district with the highest loading, it also has the 5-th lowest loading for Hancock County, GA (-1.396). The racial disparities between these two counties are particularly striking, and this difference is partially removed by the inclusion of the fixed effects. The population in Forsyth County is close to 90\% white and in Hancock County is 84\% African-American. The Forsyth school district is well-funded and uses technology extensively including tools that allow parents to monitor student assignments and grades 24 hours a day. (Detailed estimation results for both math and reading scores are available from the authors.)

\begin{figure}
\begin{center}
\centerline{\includegraphics[width=.8\textwidth]{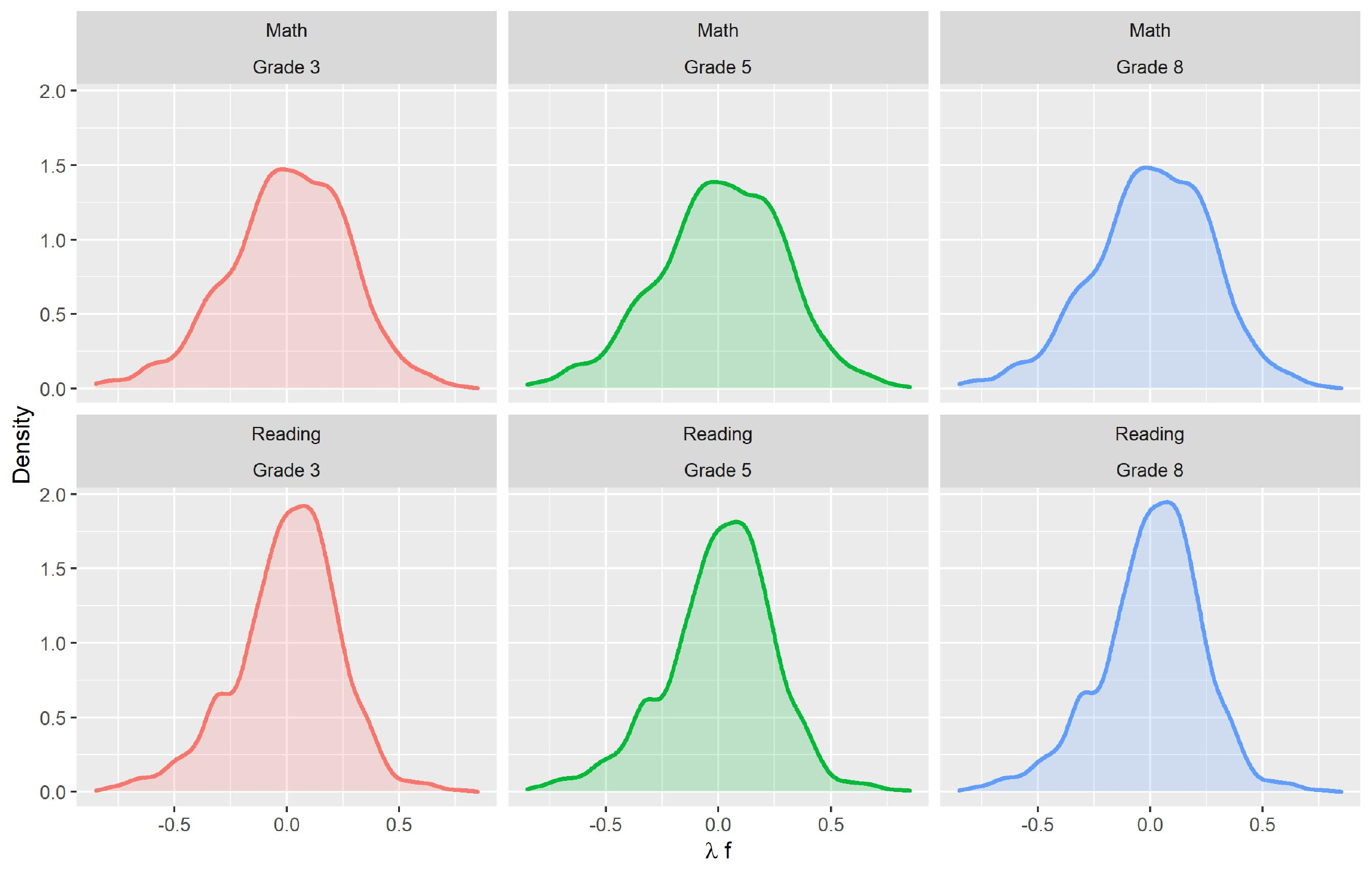}}
\caption{\emph{Changes in the distribution of district educational attainment by grade and subject.} \label{f42:results}}
\end{center}
\end{figure}

Figure \ref{f42:results} plots $\lambda_{i} f_{g}$ and allows us to evaluate the change in the distribution of unobserved district educational attainment by grade and by subject. The distributions appear to change little by grade, indicating the lack of significant differences in district educational achievement that could be attributed to the quality of education at different stages of the K-12 education system.

\subsection{Class size and college educational attainment}

Next, we use data from administrative records of the economics and finance programs at Bocconi University \citep{DeGiorgi2012} in order to estimate the following one-factor model of the effect of class size and socioeconomic class composition on educational attainment: 
\begin{equation}
y_{icj} = \bm{d}_{cj}'\bm{\alpha} + \bm{x}_{icj}' \bm{\beta} + \lambda_i f_{cj} + u_{icj}, \label{educ:eq1} 
\end{equation}
where $y_{icj}$ is the average grade of student $i$ in a class $c$ at year $j$, and $\bm{d}_{cj}$ is a vector of variables that includes class size, and measures of actual dispersion of gender and income in each class. The vector $\bm{x}_{icj}$ captures observed variables such as gender and income. In this model, $\lambda_{i}$ is associated with student motivation and ability, and the factor $f_{ct}$ is interpreted as measuring teaching quality of the course $c$ taken in year $j$. Moreover, the term $\lambda_i f_{cj}$ represents the interaction between student motivation $\lambda_{i}$ and the quality of the teacher in a class $f_{cj}$. The inclusion of interactive latent factors is considered important since it allows us to account for situations where high teaching quality can have a modest effect on the educational attainment of relatively unmotivated students, although it can dramatically affect performance among strong, motivated students. 

\vspace{3mm}
\begin{table}[hptb]
\begin{center}
\begin{tabular}{c c c c c c c c } \hline 
\multicolumn{2}{c}{}&\multicolumn{2}{c}{PCA}&\multicolumn{2}{c}{GVE}& \multicolumn{2}{c}{WGVE} \\
          &         &  $\hat{\theta}$ & $\Delta$\%  & $\hat{\theta}$ & $\Delta$\%   &  $\hat{\vartheta}$ & $\Delta$\%  \\ \hline
Course 1        &       Year 1  &       1.000   &        -      &       1.000   &       - &       0.973   &        -      \\
Course 2        &       Year 2  &       1.041   &       4.06    &       1.005   &       0.49    &       0.978   &       0.49    \\
Course 3        &       Year 3  &       1.076   &       3.35    &       1.083   &       7.74    &       1.053   &       7.74    \\ \hline
\end{tabular}
\vspace{3mm}
\caption{\emph{Estimated factors using the Bocconi's data. Principal Component Analysis is denoted by PCA, the grouped variable estimator by GVE, and the weighted grouped variable estimator by WGVE with equal weights.} \label{table:cs:factors}} 
\end{center}
\end{table}

The data set captures a rich set of covariates which are included in the specification. It includes information on course grades, background demographic and socioeconomic characteristics such us gender, family income, and pre-enrollment test scores. Additionally, the data set includes information on enrollment year, academic program, number of exams by academic year, official enrollment, official proportion of female students in each class, and official proportion of high income students in each class. We restrict our attention to students who matriculated in the 2000 academic year and took the same non-elective classes in the first three years of the program. See \citet{DeGiorgi2012} and \citet{harding2014estimating} for additional details on the data.

As in \citet{DeGiorgi2012}, we estimate $\bm{\alpha}$ and $\bm{\beta}$ in equation \eqref{educ:eq1} using instrumental variables generated by a random assignment of students into classes. Students were assigned to each class by the administration at Bocconi University, and therefore, the random assignment determine the actual class size, percentages of female students in a class and high income students in a class, which are considered to be endogenous variables. The use of the randomized assignment allows for the consistent estimation of the coefficients in equation \eqref{educ:eq1}, satisfying one of the conditions of our approach. In the second stage, we employ residuals $R_{icj} = \lambda_{i} f_{cj} + u_{icj}$, and we estimate the factors and loadings following the procedure described in Sections 2 and 3.

\begin{figure}
\begin{center}
\centerline{\includegraphics[width=0.7\textwidth]{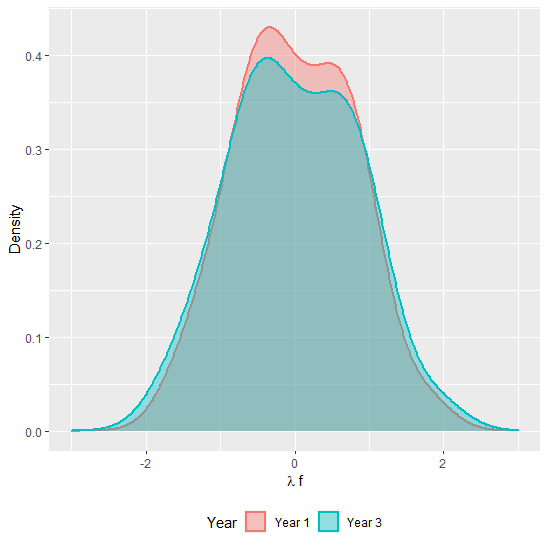}}
\caption{\emph{Changes in the distribution of student's grade by year in the program.} \label{f41:results}}
\end{center}
\end{figure}

Table \ref{table:cs:factors} shows the factors $f_{cj}$ estimated using PCA, GVE, and WGVE. Because $m_r=1$, the GVE and IV estimators are identical. While PCA suggests that teacher/course quality $f_{cj}$ does seem to improve linearly over years, the WGVE suggests that the quality of courses improves mainly in the third year of the program. In Figure \ref{f41:results}, we estimate the distribution of $\hat{\lambda}_i \hat{f}_{cj}$ by years in the program. It is interesting to see that the middle and upper tail of the distribution changes over time, and by the third year, the conditional distribution of the average grade becomes more dispersed. This finding suggest that the students who remained in the program became more heterogeneous and the latent abilities of the high-performing students improved over time. 

\subsection{Vouchers and educational attainment}\label{Colombia}

Lastly, we investigate the impact of an educational voucher program that provided opportunities for students to attend private schools. During past decades, numerous educational voucher programs were adopted in the U.S. and Latin America. The empirical literature focused on the evaluation of the effect of the program on observable outcomes \citep[see, e.g.][]{jAngrist2002,jAngrist2006,LAMARCHE2011278}, but the effect of such programs on latent variables such as cognitive ability of students is unknown. In this example, we illustrate the use of our estimation approach using data from \citet{jAngrist2002} concerning a 1991 program in Colombia. The vouchers were assigned using lotteries, and they were renewable as long as the students maintained satisfactory academic progress.

We estimate the following factor-augmented linear panel data model:
\begin{equation}
y_{is} = \alpha d_i + \bm{x}_{is}' \bm{\beta} + \lambda_i f_s + u_{is}, 
\end{equation}
where $y_{is}$ is student's $i$ test score in subject $s$ and $d_i$ indicates treatment status (i.e., whether student $i$ won the lottery). The parameter $\alpha$ is the mean treatment effect of the program. The vector of independent variables is denoted by $\bm{x}_{ij}$ and the error term by $u_{is}$. The loading $\lambda_i$ measures the student's intrinsic ability or effort that also drives performance in the three subjects, and the variable $f_s$ is a subject specific effect that impacts student achievement.

\vspace{3mm}
\begin{table}[hptb]
\begin{center}
\begin{tabular}{c c c c c  c c}  \hline 
\multicolumn{1}{c}{} & \multicolumn{3}{c}{Control} & \multicolumn{3}{c}{Treatment} \\ 
      & PCA               &  GVE          &  WGVE               & PCA    &       GVE     &       WGVE \\ \hline  
Mathematics     &        1.000  &        1.000  &        0.679  &        1.000   &        1.000  &        0.936  \\
Reading           &      1.575  &        2.679  &        1.850  &        0.970   &        2.522  &        1.184  \\
Writing     &    1.437  &        1.468  &        1.013  &        0.866  &        1.960   &        0.920  \\  \hline
\end{tabular}
\vspace{3mm}
\caption{\emph{Factors estimated from the Colombian voucher data (PACES). Principal Component Analysis is denoted by PCA, the grouped variable estimator by GVE, and the weighted grouped variable estimator by WGVE with equal weights.} \label{table:voucher:factors}} 
\end{center}
\end{table}

We use data on 284 students who took tests in mathematics, reading and writing. These tests were taken three years after the vouchers were distributed. To facilitate the comparison among subjects, the test scores are in standard deviation units. In addition to an indicator variable for whether the student won a voucher, we use the following independent variables: site dummies, strata indicators for whether the student lives in a neighborhood ranked on a scale of 1-6 from poorest to richest, an indicator for whether the interview was done by a house visit since telephones were used in the majority of the interviews, gender, age, and parents' schooling. We also include an indicator for survey form, because \citet{jAngrist2002} data also incorporate responses obtained from a pilot survey designed to test questions and interviewing strategies. 

Table \ref{table:voucher:factors} presents the factors for Mathematics, Reading and Writing. We estimate separately $f_s$ for students in the control and treatment group, to measure whether these factors differ by treatment status. The table also presents results using the estimation approaches introduced in this paper. The results for Mathematics and Writing in the control group are qualitatively similar when using PCA, GVE or WGVE. In contrast, WGVE estimates significant gains in Mathematics for the treatment group. The results do not seem to suggest improvements in the other subjects resulting from the treatment. 

\begin{figure}
\begin{center}
\centerline{\includegraphics[width=.7\textwidth]{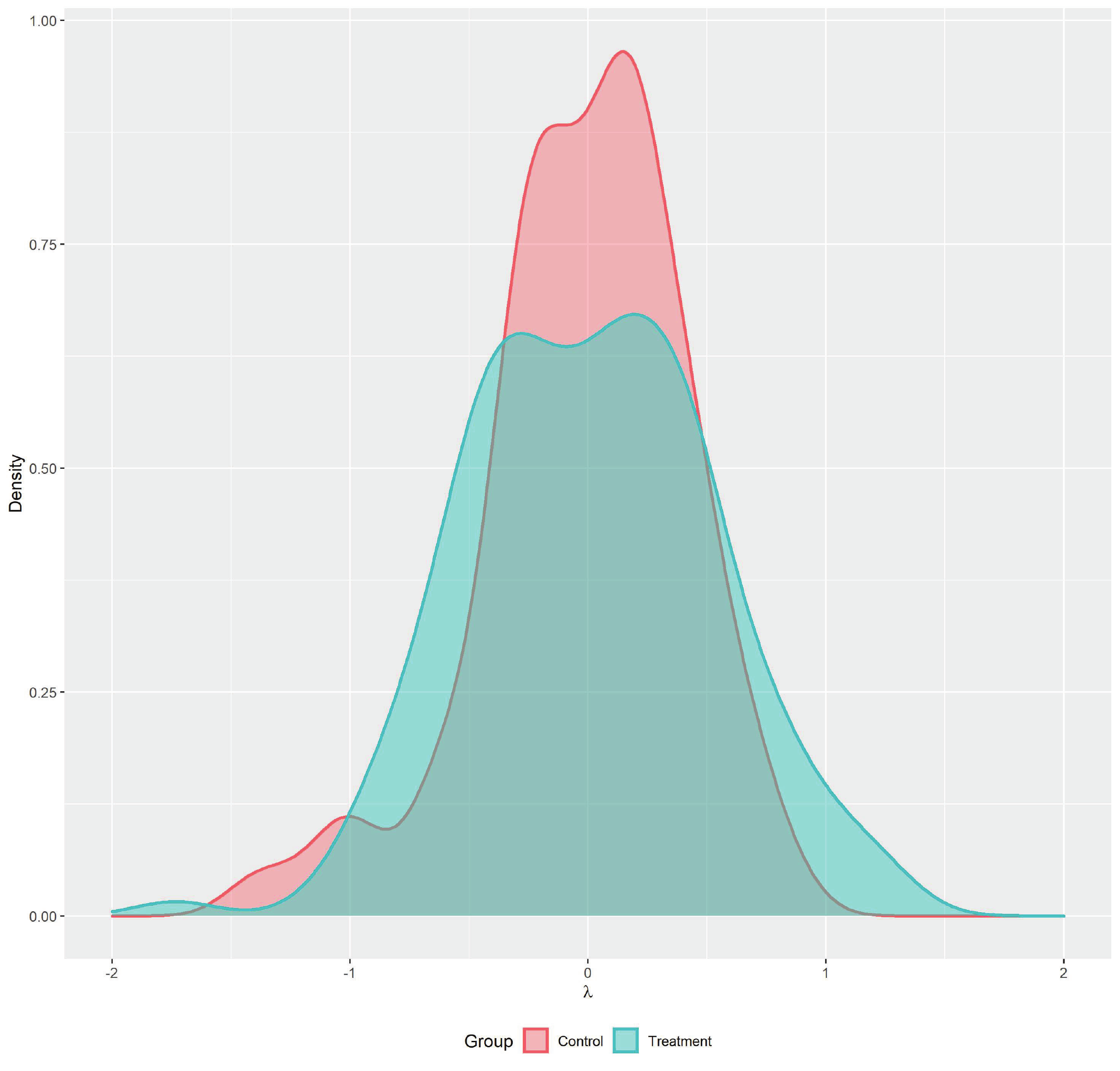}}
\caption{\emph{The effect of winning a voucher for private schooling on latent ability.} \label{f52:results}}
\end{center}
\end{figure}

Lastly, to summarize the impact of the voucher program on student achievement, we can evaluate the factor structure in our model of academic achievement. Figure \ref{f52:results} plots the distribution of student's latent ability by treatment status. The figure reveals that the educational policy implemented in Colombia improved latent cognitive outcomes of low-performing students and high-performing students, while increasing the gap between strong and weak students.  We also measure a difference of 0.016 between the values of $\lambda_i$ for students in the treatment and control groups.

\section{Conclusions and Discussion}

In this paper, we investigated the estimation of factor-augmented linear models using internally generated instruments in the spirit of \cite{madansky1964instrumental}, while addressing challenges such as the potentially large number of equally valid instruments. Given that many normalizations are possible for the identification of such a model, we explore the advantages of creating linear combinations of IV estimators, which leads to efficiency improvements.

While the proposed approach is computationally intensive and identification relies on correctly specifying the dependence between the latent factors and the error term, it nevertheless leads to a simple approach to estimating the latent factors in linear models. Further research may involve relaxing the identification assumptions to more general cases and to the extension of approximate factor models.

\appendix

\begin{appendices}

\section{Proofs}\label{app:proofs}

\begin{proof}[Proof of Theorem \ref{th:HL_version}.]

Our estimator is a 2SLS estimator that fits into the framework of \citet{bHansen2019}'s Section 5 considering $R_{n}$ equal to the identity matrix and  $d_n = 1$. Our cross-sectional units $i$ play the role of cluster units $g$ in their paper, and our $j \in A_0$ plays the role of their within-cluster observations. Thus, we proceed by showing that the instruments we propose are valid, and then verifying the conditions of Theorems 8 and 9 in \citet{bHansen2019}.


Condition (d) ensures that the estimand is well-defined, and that our condition (b) implies instrument validity, i.e. $ E\left[ \bm{Z}_{iA}' \bm{v}_{iA_0} \right] = \bm{0}$. To see this, recall that instruments are (averages of) regressors in $A_0 \cup A_J \cup B_J$ and (averages of) dependent variables in $B_J$, and that the error term $v_{ij}$ for $j \in A_0$ involves error terms from groups $j$ and partition $A_J$,
\begin{equation}
    v_{ij} = u_{ij} - \bm{\theta}_{j}' \overline{\bm{u}}_{iA_J}. \label{append:err}
\end{equation}
That any regressor is exogenous follows directly from the final part of assumption (b). To see that condition (b) is sufficient for the exogeneity of the instruments, let $h\in B_J$ and $j \in A_0$. By definition, we have a vector $\bar{\bm{y}}_{iB_J} = (\bar{y}_{iB_J,1}, \hdots, \bar{y}_{iB_J,r})'$ of $r$ averages with first element defined as $\bar{y}_{iB_J,1} = m_r^{-1} \sum_{h = 1}^{m_r} y_{ih}$, second element $\bar{y}_{iB_J,2} = m_r^{-1} \sum_{h = m_r+1}^{2 m_r} y_{ih}$, etc. Then, $$E (\bar{y}_{iB_J,k}  v_{ij} ) =  \frac{1}{m_r} \sum_{h = k m_r+1}^{k  m_r} E (y_{ih} v_{ij}).$$ 
Using Assumption (a) and equation \eqref{append:err}, for each $h$, we have
\begin{align*}
E\left(y_{ih} v_{ij}\right) &= E\left(\left(\bm{x}_{ih}'\bm{\beta}+\bm{\lambda}_{i}'\bm{f}_{h} + u_{ih} \right) \left((u_{ij} - \bm{\theta}_{j}' \overline{\bm{u}}_{iA_J} \right)\right) \\
&= E\left(\bm{\lambda}_{i}'\bm{f}_{h} u_{ij}\right) + E\left(\bm{x}_{ih}'\bm{\beta} u_{ij} \right) + E\left(u_{ih} u_{ij} \right) \\
&- E\left(\bm{\lambda}_{i}'\bm{f}_{h} \bm{\theta}_{j}' \overline{\bm{u}}_{iA_J} \right) \\
&- E\left(\bm{x}_{ih}'\bm{\beta} \bm{\theta}_{j}' \overline{\bm{u}}_{iA_J} \right) \\
&- E\left(u_{ih} \bm{\theta}_{j}' \overline{\bm{u}}_{iA_J}\right) \\
&= 0,
\end{align*}
where the terms in the first line after the second equality are zero because of the first, second, and third component of Assumption (b), and because $A_J \cap B_J = \emptyset$ and that $\bm{\beta}$ and $\bm{f_h}$ are non-random. The other terms are treated similarly.

It remains to verify the conditions in Theorem 9 of \citet{bHansen2019} (which imply those in Theorem 8). Their Assumption 2 holds by construction because our panel is balanced and because $m_{A_0}$ is fixed. The full rank and minimum eigenvalue assumptions on $\bm{Q}_N$, $\bm{W}_N$, and $\bm{\Omega}_N$ are directly assumed in the statement of our result, via condition (e).

We check that for some $s>2$, the dependent variable, regressors, and instrumental variables have bounded $2s$th moments. For the dependent variable, this is assumed in our condition (c). For the regressors and instruments, note that they are either (averages of) dependent variables, or (averages of) regressors. If they are not averages, our assumption (c) directly assumes that the $2s$th moment is bounded. If they are averages (over $A_J$ or $B_J$), it follows from our Assumption
(c) and a $c_{r}$-inequality. To see this,
\begin{align*}
\sup_{i}E\left|\overline{\bm{y}}_{iA_J}\right|^{2s} & =\sup_{i}E\left|\frac{1}{m_{A_{J}}}\sum_{k\in A_{J}}y_{ik}\right|^{2s}\\
 & \leq\sup_{i}m_{A_{J}}^{2s-1}\sum_{k\in A_{J}}E\left|\frac{1}{m_{A_{J}}}y_{ik}\right|^{2s}\\
 & =\sup_{i}\frac{1}{m_{A_{J}}}\sum_{k\in A_{J}}E\left|y_{ik}\right|^{2s}\\
 & \leq\sup_{i,k}E\left|y_{ik}\right|^{2s}<\infty
\end{align*}
where the first equality is the definition of $\overline{\bm{y}}_{iA_J}$,
the second inequality is the $c_{r}$ inequality, the third equality moves out the $m_{A_{J}}$ from the absolute value and cancels it against the term in front of the sum; the fourth inequality uses that the supremum of the moment over all periods is at least as big as the moment in any given time periods, and the boundedness follows from our assumption (c). The argument for $\overline{\bm{y}}_{iB_J}$ and for each element of averaged regressors is almost identical. Because all elements of the matrices of regressors and instrumental variables have bounded $2s$ moments, so do the matrices.

\end{proof}

\begin{proof}[Proof of Theorem \ref{thm:MGVE-fixed-T}]
We begin establishing consistency, and then, in the second part of the proof, we obtain the asymptotic distribution of the WGVE. Using equations \eqref{eq:newNST} and \eqref{MGVE}, we have 
\begin{equation*}
\widehat{\bm{\vartheta}}_{A_0} - \bm{\vartheta}_{A_0} = \sum_{q=1}^{Q_J} W_q \left( \widehat{\bm{\theta}}_{A_0,(q)} - \bm{\theta}_{A_0,(q)} \right).
\end{equation*}
It follows then that
\begin{equation}
\| \widehat{\bm{\vartheta}}_{A_0} - \bm{\vartheta}_{A_0} \| \leq \sum_{q=1}^{Q_J} \| W_q  \|  \| \widehat{\bm{\theta}}_{A_0,(q)} - \bm{\theta}_{A_0,(q)} \|.
\end{equation}
To show that the estimator is consistent, we need to show that $\| \widehat{\bm{\theta}}_{A_0,(q)} - \bm{\theta}_{A_0,(q)} \| = o_p(1)$, which can be established using Theorem \ref{th:HL_version} under Assumption (i). The result follows since $Q_J$ is fixed and the weights are bounded by condition (iii).

To show asymptotic normality, we need to establish the limiting distribution of
\begin{eqnarray}
\sqrt{N} \left( \widehat{\bm{\vartheta}}_{A_0} - \bm{\vartheta}_{A_0} \right)  & = & \sum_{q=1}^{Q_J} W_q  \sqrt{N} \left( \widehat{\bm{\theta}}_{A_0,(q)} - \bm{\theta}_{A_0,(q)} \right) \nonumber \\ 
& = & \sum_{q=1}^{Q_J} W_q  \bm{\Sigma}_{N,(q)}^{1/2}  \bm{\Sigma}_{N,(q)}^{-1/2} \sqrt{N} \left( \widehat{\bm{\theta}}_{A_0,(q)} - \bm{\theta}_{A_0,(q)} \right) \nonumber \\ 
 & = &  \sum_{q=1}^{Q_J} W_q \bm{\Sigma}_{N,(q)}^{1/2} \hat{\bm{\xi}}_{N,(q)}, \label{eq:MGVE-AN-fixedT}
\end{eqnarray}
where $\hat{\bm{\xi}}_{N,(q)} = \bm{\Sigma}_{N,(q)}^{-1/2} \sqrt{N} ( \widehat{\bm{\theta}}_{A_0,(q)} - \bm{\theta}_{A_0,(q)})$ is an asymptotically normal random variable by Theorem \ref{th:HL_version}. The results follows since the right hand side of \eqref{eq:MGVE-AN-fixedT} includes a weighted sum of a finite number of asymptotically normal random variables and those weights are bounded by assumptions (ii) and (iii). 
\end{proof}

\begin{proof}[Proof of Theorem \ref{T3}]
The proof has three parts. First, we demonstrate the consistency of the estimator defined as the solution of \eqref{GVE3}. In the second part of the proof, we show that the WGVE as defined in \eqref{MGVE-largeT} is also consistent. Lastly, we  establish the asymptotic normality of the estimator. 

The proof for the consistency of \eqref{GVE3} follows directly from \citet*{Belloni2012}, and it requires to verify that our assumptions satisfy conditions AS and CF in their paper. Equation \eqref{AS} is similar to the approximate sparsity (AS) condition in \citet*{Belloni2012}, which imposes a uniform upper bound $s$ for the number of variables approximating conditional expectation functions. In terms of the behavior of the population Gram matrix, we verify that condition RE in \citet*{Belloni2012} is Assumption (b). 

Moreover, Assumption (c) states a set of sufficient conditions that are comparable to Condition RF in \citet*{Belloni2012}. In our factor model for normalization $q$, we have that $$Cov\left( R_{ij,(q)}, R_{kj,(q)} \right) =  Cov \left( \left( \bm{f}_{j,(q)}' \bm{\lambda}_{i} + u_{ij,(q)} \right), \left( \bm{f}_{j,(q)}' \bm{\lambda}_{k} + u_{kj,(q)} \right) \right) = 0$$ because $\bm{f}_{j,(q)}$ is a parameter, $\bm{\lambda}_{i}$ and $u_{is}$ are independent and $u_{is}$ is independent within $i$ by Assumption (b) in Theorem \ref{th:HL_version}. The variable $R_{ij,(q)}$ is bounded under Assumption (c). Moreover,
\begin{eqnarray*}
\tilde{R}_{il,(q)} & = & R_{il,(q)} - \frac{1}{N} \sum_{i=1}^N E\left( R_{il,(q)} \right) = R_{il,(q)} - \frac{1}{N} \sum_{i=1}^N E\left( \bm{R}_{i,B_J,(q)} \bm{\pi}_{l,(q)} + \epsilon_{il,(q)} \right) \\
                                                                                & = & R_{il,(q)} - \bar{\bm{R}}_{B_J,(q)} \bm{\pi}_{l,(q)} = \bm{\pi}_{l,(q)}' \left( \bm{R}_{i,B_{J},(q)} - \bar{\bm{R}}_{B_{J},(q)} \right) + \epsilon_{il,(q)},
\end{eqnarray*}
where $\bar{\bm{R}}_{B_{J},(q)} = \sum_{i=1}^N E\left( \bm{R}_{iB_{J},(q)} \right)$. The independence of the vector of normalized endogenous variables follows from the first part of Assumption (c) and Assumption (a) in Theorem \ref{th:HL_version}.

We now follow closely Theorem 1 and Lemma 1 in \citet*{linton2016}, and here we focus on the main differences. First, we write
\begin{equation}
\| \hat{\bm{\vartheta}}_j - \bm{\vartheta}_j \| \leq \sum_{q \in \mathcal{Q}_J^\ast} \| W_{q} \| \times \max_{q \in \mathcal{Q}_J^\ast} \| \hat{\bm{\theta}}_{j,(q)} - \bm{\theta}_{j,(q)} \| + \sum_{q = Q_J^\ast + 1}^{Q_J} \| W_{q} \| \Delta,
\end{equation}
where $\Delta$ is the radius of the compact set $\bm{\Theta}$. By Assumption (a), the last term does $o_p(1)$ as $J \to \infty$, so we concentrate in the first term. To show that the estimator $\hat{\bm{\vartheta}}_j$, which based on a linear combination of consistent estimators, is consistent, we need to show that $\max_{q \in \mathcal{Q}_J^\ast} \| \hat{\bm{\theta}}_{j,(q)} - \bm{\theta}_{j,(q)} \| = o_p(1)$ as $N$ and $J$ go jointly to $\infty$, because the weights are bounded by assumption (a). 

By Assumption (d), if  $\max_{q \in \mathcal{Q}_J^\ast} \| \hat{\bm{\theta}}_{j,(q)} - \bm{\theta}_{j,(q)} \|  > \eta$, we have that $\| G_q (\hat{\bm{\theta}}_{j}) \| \geq \epsilon_N(\eta)$ for some $q$. Therefore,
\begin{equation}
\mbox{Pr} \left(  \max_{q \in \mathcal{Q}_J^\ast} \| \hat{\bm{\theta}}_{j,(q)} - \bm{\theta}_{j} \|  > \eta \right) \leq \mbox{Pr} \left( \max_{q \in \mathcal{Q}_J^\ast} \| G_q (\hat{\bm{\theta}}_{j}) \|  \geq  \epsilon_N(\eta) \right).
\end{equation}

For $\epsilon_N(\eta)>0$, 
\begin{eqnarray*}
\max_{q \in \mathcal{Q}_J^\ast} \| G_q (\hat{\bm{\theta}}_{j}) \| & \leq & \max_{q \in \mathcal{Q}_J^\ast} \|  G_q(\hat{\bm{\theta}}_{j}) - G_{N,(q)}(\hat{\bm{\theta}}_{j)}) \| + \max_{q \in \mathcal{Q}_J^\ast} \|  G_{N,(q)}(\hat{\bm{\theta}}_{j}) \| \\
& \leq & \max_{q \in \mathcal{Q}_J^\ast}  \sup_{ \theta_j \in \Theta_j}  \|  G_q(\bm{\theta}_{j}) - G_{N,(q)}( \bm{\theta}_{j}) \| + \max_{q \in \mathcal{Q}_J^\ast} \|  G_{N,(q)}(\hat{\bm{\theta}}_{j}) \| \\
& = & o_p(\alpha_{N}) + \max_{q \in \mathcal{Q}_J^\ast} \|  G_{N,(q)}(\hat{\bm{\theta}}_{j}) \|,
\end{eqnarray*}
by Assumption (e). The consistency result follows by definition of $\hat{\bm{\theta}}_{j,(q)}$ and $\epsilon_{N}(\eta) \to 0$ as $N$ and $J$ go to infinity. 

Under similar conditions, Theorem 3 in \citet*{Belloni2012} demonstrate that the estimator $\hat{\bm{\theta}}_{j,(q)}$ is asymptotically normal, and this implies that the weighted sum over $Q_J^\ast$ is asymptotically normal under conditions (a) and (c). Because the binomial coefficient $Q_J \to \infty$ rapidly as $J \to \infty$, but the standarized binomial coefficient $Q_J^\ast \to C$, the estimator is a weighted average of a finite number of normalizations. It follows that as both $N$ and $J \to \infty$ under the rates in condition (c) established for the Lasso-type estimator in \citet*{Belloni2012}, $\widehat{\bm{\vartheta}}_{j}$ is asymptotically normal with covariance matrix, 
\begin{equation*}
\bm{\Omega}_j = \lim_{N \to \infty} N \sum_{q=1}^{C} \sum_{l=1}^{C} W_{q} E\left( \left(\widehat{\bm{\theta}}_{j,(q)}-\bm{\theta}_{j,(q)}\right) \left(\widehat{\bm{\theta}}_{j,(l)}-\bm{\theta}_{j,(l)}\right)' \right) W_{l}'.
\end{equation*}

\end{proof}

\end{appendices}

\bibliographystyle{econometrica}
\bibliography{hlm}

\end{document}